\documentclass{IEEEtran}
\usepackage{cite}
\usepackage{amsmath,amssymb,amsfonts}
\usepackage{graphicx}
\usepackage{textcomp}

\usepackage{subcaption}
\usepackage{amsthm}
\usepackage{cite}
\usepackage{color}
\usepackage{comment}
\usepackage[overload]{empheq}
\usepackage{algorithm,algcompatible}% http://ctan.org/pkg/algorithm
\usepackage{xspace}% http://ctan.org/pkg/xspace

\makeatletter
\newcommand{\ALOOP}[1]{\ALC@it\algorithmicloop\ #1%
	\begin{ALC@loop}}
	\newcommand{\ENDALOOP}{\end{ALC@loop}\ALC@it\algorithmicendloop}

\algnewcommand\INITIALIZE{\item[\textbf{Initialization:}]}%
\makeatother

\newtheorem{thm}{\bf Theorem}
\newtheorem{prop}{\bf Proposition}
\newtheorem{lem}{\bf Lemma}
\newtheorem{assmpt}{\bf Assumption}
\newtheorem{defn}{\bf Definition}

\newtheorem{rem1}{\bf Remark}

\newcounter{Examp}

\newenvironment{rem}{\begin{rem1}\em}{\hfill$\Diamond$\end{rem1}}
\newcommand{\C}{\ensuremath{{\mathbb C}}}
\newcommand{\R}{\ensuremath{{\mathbb R}}}

\newcommand{\N}{\ensuremath{{\mathbb N}}}

\newcommand{\DD}{{\mathcal D}}

\newcommand{\MM}{{\mathcal M}}

\newcommand{\NN}{{\mathcal N}}

\newcommand{\RR}{{\mathcal R}}

\newcommand{\VV}{{\mathcal V}}
\newcommand{\PP}{{\mathcal P}}
\newcommand{\OO}{{\mathcal O}}
\newcommand{\EE}{{\mathcal E}}
\newcommand{\CC}{{\mathcal C}}

\newcommand{\supp}{\mathsf{supp}}

\newcommand{\cospark}{\mathsf{cospark}}

\newcommand{\ii}{\mathsf{i}}
\newcommand{\jj}{\mathsf{j}}
\newcommand{\mm}{\mathsf{m}}
\newcommand{\nn}{\mathsf{n}}
\newcommand{\pp}{\mathsf{p}}
\newcommand{\qq}{\mathsf{q}}
\newcommand{\rr}{\mathsf{r}}
\newcommand{\np}{\mathsf{np}}

\DeclareMathOperator*{\argmin}{arg\,min}

\def\BibTeX{{\rm B\kern-.05em{\sc i\kern-.025em b}\kern-.08em
    T\kern-.1667em\lower.7ex\hbox{E}\kern-.125emX}}

\begin{document}
	
	% vetical space between paragraph and equation 	
	\setlength{\abovedisplayskip}{5pt}
	\setlength{\belowdisplayskip}{5pt}
	
\title{On Redundant Observability: From Security Index to Attack Detection and Resilient State Estimation}
\author{Chanhwa Lee, Hyungbo Shim, and Yongsoon Eun
\thanks{A preliminary version of this paper was presented at the 14th European Control Conference (ECC'15) as \cite{Lee15}, where a theoretical derivation of resilient state estimation for a continuous-time system was mainly discussed without any concrete structure, detailed operation algorithm, or relationship with the security index and attack detection.}
\thanks{This work was supported in part by Institute for Information \& communications Technology Promotion (IITP) grant funded by the Korea government (MSIP) (2014-0-00065, Resilient Cyber-Physical Systems Research) and in par by Global Research Laboratory Program through the National Research Foundation of Korea (NRF) funded by the Ministry of Science and ICT (NRF-2013K1A1A2A02078326).}
\thanks{C.~Lee is with Research \& Development Division, Hyundai Motor Company, Korea (e-mail: chanhwa.lee@gmail.com). }
\thanks{H.~Shim is with ASRI, Department of Electrical and Computer Engineering, Seoul National University, Korea (e-mail: hshim@snu.ac.kr). }
\thanks{Y.~Eun is with Department of Information \& Communication Engineering, DGIST, Korea (e-mail: yeun@dgist.ac.kr).}}

\maketitle

\begin{abstract}
The security of control systems under sensor attacks is investigated. 
Redundant observability is introduced, explaining existing security notions including the security index, attack detectability, and observability under attacks. 
Equivalent conditions between redundant observability and existing notions are presented. 
Based on a bank of partial observers utilizing Kalman decomposition and a decoder exploiting redundancy, an estimator design algorithm is proposed enhancing the resilience of control systems. 
This scheme substantially improves computational efficiency utilizing far less memory. 

\end{abstract}


\noindent
{\em Notation:} The subset of natural numbers, $\left\{1, 2, \cdots, \pp \right\} \subset \N$, is denoted by $[\pp]$. The cardinality of a set $S$ is denoted by $|S|$ and the support of a vector $y \in \C^\pp$ is defined as $\supp(y):= \left\{ \ii \in [\pp] : y_\ii \neq 0 \right\}$ where $y_\ii$ is the $\ii$-th element of $y$. The cardinality of $\supp(y)$ defines the $\ell_0$ norm of a vector $y$, i.e., $\|y\|_0:= |\supp(y)|$.
A vector $y$ is said to be $\qq$-sparse if $\|y\|_0 \leq \qq$.
The set $\Sigma_\qq := \left\{ y \in \C^\pp : \|y\|_0 \leq \qq \right\}$ denotes the set of all $\qq$-sparse vectors.
The 2-norm of a vector $y$ is defined as $\|y\|_2:= \sqrt{y^* y}$ where $y^*$ is the Hermitian of $y$.

Assume that a vector $y \in \C^\pp$ and a subset $\Lambda \subset [\pp]$ of indices are given. 
We use the notation  $y_\Lambda\in \C^\pp$ to denote that $y_\Lambda$ is obtained by setting the elements of $y$ indexed by 
$\Lambda^c := [\pp]\setminus\Lambda = \left\{ \ii \in [\pp] : \ii \notin \Lambda \right\}$
to zero. Similar notation is used for a matrix $C \in \R^{\pp\times \nn}$. The matrix obtained by setting the rows of $C$ indexed by $\Lambda^c$ to zero, is denoted as $C_\Lambda \in \R^{\pp\times \nn}$. 
Sometimes the notation will be slightly modified to $y_\Lambda^\pi \in \C^{|\Lambda|}$ (or $C_\Lambda^\pi \in \R^{|\Lambda|\times \nn}$), which denotes the vector $y$ (or the matrix $C$) whose elements (or rows) not corresponding to the index set $\Lambda$ are actually eliminated.

For a given index $\ii \in [\pp]$, the index set $\Gamma_\ii^\nn \subset [\np]$ represents $\left\{ \nn(\ii-1)+1, \nn(\ii-1)+2, \cdots, \nn\ii \right\}$. Similarly, for a given index set $\Lambda \subset [\pp]$, the index set $\Lambda^\nn \subset [\np]$ denotes $ \bigcup_{\ii \in \Lambda} \Gamma_\ii^\nn$.
A vector $z \in \C^{\np}$ of length $\np$ can be split into $\pp$ column vectors of length $\nn$, i.e., $z =\left[ z_1^{\nn\top} ~ z_2^{\nn\top} ~ \cdots ~ z_\pp^{\nn\top} \right]^\top \in \C^{\np}$, where $z_\ii^\nn \in \C^\nn$ represents the $\ii$-th split column vector of length $\nn$ in $z$.
Then we call $z$ an $\nn$-stacked vector.
With the index set $\Gamma_\ii^\nn$ defined above, it follows that $z_\ii^\nn = z_{\Gamma_\ii^\nn}^\pi \in \C^\nn$. 
The ($\nn$-stacked) support of $z \in \C^\np$ is defined as $\supp^\nn(z):= \left\{ \ii \in [\pp] : z_\ii^\nn \neq 0_{\nn\times 1} \right\}$
and its cardinality defines the ($\nn$-stacked) $\ell_0$ norm of $z$, i.e.,
$\|z\|_{0^\nn}:= |\supp^\nn(z)|.$
Similarly to the usual vector case, an $\nn$-stacked vector $z$ is said to be ($\nn$-stacked) $\qq$-sparse when it holds that $\|z\|_{0^\nn} \leq \qq$, and the set
$\Sigma_\qq^\nn := \left\{ z  \in \C^\np : \|z\|_{0^\nn} \leq \qq \right\}$
denotes the set of all ($\nn$-stacked) $\qq$-sparse vectors.

For a matrix $C \in \R^{\pp\times\nn}$, the cospark of $C$ is defined as $\displaystyle\cospark(C) := \min_{x\in\R^\nn,\!~ x\neq 0_{\nn\times 1}} \|C x\|_0$
and the ($\nn$-stacked) cospark of a matrix $\Phi \in \R^{\np\times\nn}$ is similarly defined as
$\displaystyle\cospark^\nn(\Phi) := \min_{x\in\R^\nn,\!~ x\neq 0_{\nn\times 1}} \|\Phi x\|_{0^\nn}.$
Subspaces $\RR(C)$ and $\NN(C)$ denote the range space and the null space of $C$, respectively.
The induced matrix 2-norm of a matrix $C$ is defined as $\|C\|_2:= \sqrt{\lambda_{\max} \left( C^\top C \right) } = \sigma_{\max}(C)$ where $\lambda_{\max}(\cdot)$ and $\sigma_{\max}(\cdot)$ denote the maximum eigenvalue and the maximum singular value, respectively. 
In addition, $\sigma_{\min}(\cdot)$ is used to denote the minimum singular value and $C^{\dagger}$ is the pseudoinverse of $C$.
Finally, the set of normalized eigenvectors of a square matrix $A \in \R^{\nn\times\nn}$ is denoted as $\VV(A) := \left\{ v\in \C^\nn : A v = \lambda v {\rm~for~some~} \lambda\in\C,~\|v\|_2 =  1  \right\}.$

%%%%%%%%%%%%%%%%%%%%%%
\section{Introduction}
%%%%%%%%%%%%%%%%%%%%%%

The reliability of systems in various circumstances is one of the main concerns for control engineers, and thus robust and fault-tolerant control methods have been developed to cope with model uncertainties, external disturbances, and failures in system components. 
Recently, new threats or vulnerabilities caused by malicious attacks have been reported as advances in computers and communications increase the connectivity and openness of systems \cite{Langner11}.
Therefore, the resilience of control systems to attack has become a critical system design consideration \cite{Mo12, Pasqualetti13, Sandberg15} and the security problems of the system whose measurements are compromised by adversaries have been studied actively because sensors are one of the most vulnerable points for the security of control systems \cite{Liu11,Hendrickx14,Chen15,Fawzi14,Pajic14,Shouky14,Shou15ACC,Shouky16,Chong15,Jeon16}.

In this paper, we consider a discrete-time linear time invariant (LTI) system under sensor attacks written as
\begin{align}[left = \PP: \empheqlbrace\,] 
\begin{split} \label{eq:plant}
& x(k+1) = A x(k) + B u(k) + d(k) \\
& \bar y(k) = y(k) + a(k) = C x(k) + n(k) + a(k),
\end{split}
\end{align}
where $x \in \R^\nn$ denotes the state variables, $u \in \R^\mm$ denotes the control inputs, $y \in \R^\pp$ denotes the attack-free sensor outputs, and  $\bar y \in \R^\pp$ denotes the measurement data under attack signals.
The dynamics are disrupted by the process disturbance $d \in \R^\nn$ and sensors are corrupted by the sensor attack $a \in \R^\pp$ as well as the measurement noise $n \in \R^\pp$.
There is a total of $\pp$ sensors that measure the system outputs and the $\ii$-th measurement data at time $k$ is denoted by $\bar y_{\ii}(k) = c_{\ii} x(k) + n_{\ii} (k) + a_{\ii} (k)$, where $c_\ii$ is the $\ii$-th row of $C$. 
It is assumed that the disturbances/noises are uniformly bounded, and the attacks can compromise up to $\qq$ out of $\pp$ sensor outputs, as follows. 

% % % Assumption: ass:bdd
\begin{assmpt} \label{ass:bdd}
	The process disturbance $d$ and each measurement noise $n_\ii$ are uniformly bounded, i.e.,
	\begin{equation*} %\label{eq:max}
	\|d(k)\|_2 \le d_{\max}, ~~ \|n_\ii (k)\|_2 \le n_{\max}, ~~ ^\forall k \ge 0, ~~ ^\forall \ii \in [\pp].\tag*{$\Diamond$}
	\end{equation*}
\end{assmpt}

% % % Assumption: ass:sparse
\begin{assmpt} \label{ass:sparse}
	There exist at least $\pp-\qq$ sensors that are not attacked for all $k \ge 0$, i.e., 
	\begin{equation*} %\label{eq:max}
	\left| \left\{ \ii \in [\pp] : a_\ii (k)=0, ~^\forall k \ge 0 \right\} \right| \ge \pp-\qq.\tag*{$\Diamond$}
	\end{equation*}
\end{assmpt}

\begin{figure}
	\centering
	\includegraphics[width=0.8\columnwidth]{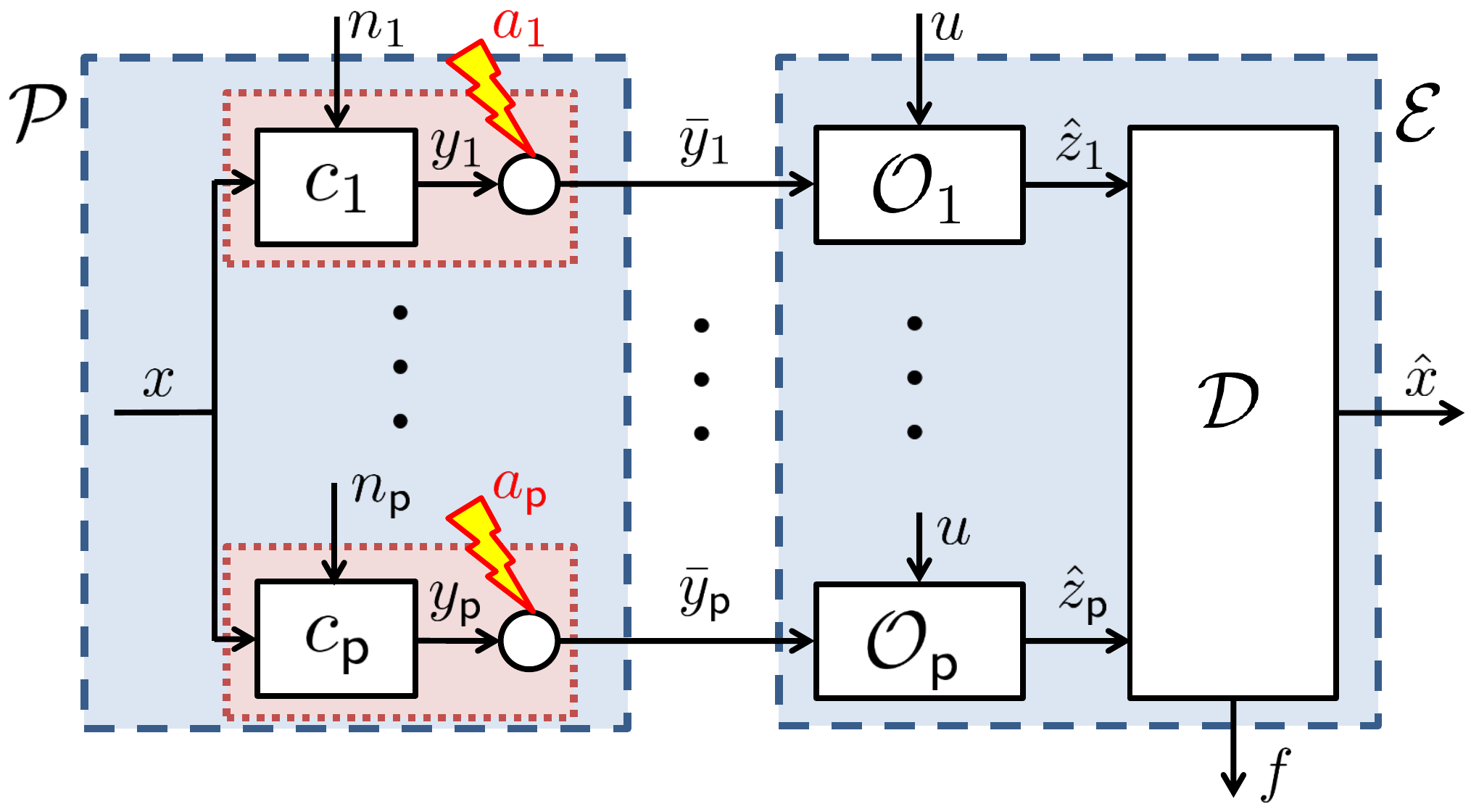}
	\caption{Configuration of the plant $\PP$ and the state estimator $\EE$.
	}\label{fig:configuration}
\end{figure}

The primary objective of this paper is to design an estimator $\EE$ that detects the attacked sensors and estimates the state $x(k)$ of the given system $\PP$ under Assumptions \ref{ass:bdd} and \ref{ass:sparse}. 
To this end, we first characterize the conditions under which the attack can be detected and the state of $\PP$ can be estimated correctly. 
Then, we construct an attack-resilient estimator $\EE$ that is composed of $\pp$ partial observers\footnote{In this paper, the terms ``observer'' and ``estimator'' are used to indicate the block of $\OO_\ii$ and $\EE$ in Fig.~\ref{fig:configuration}, respectively. That is, two terminologies should be distinguished.} $\OO_\ii$ and a decoder $\DD$ as shown in Fig.~\ref{fig:configuration}.
In other words, the characterization of observability with unknown signal $a$ and construction of the estimator are the two main topics of this paper.\footnote{From the control theoretic perspective, strong observability \cite{Basile69} and unknown input observer (UIO)-based fault estimation \cite{Gao16} may be closely related to the subject of interest here.
	If we consider the output equation $\bar y(k) = C x(k) + n(k) + I_\Lambda a(k)$ (instead of imposing Assumption \ref{ass:sparse} on $a$ in \eqref{eq:plant}) where $I \in \R^{\pp \times \pp}$ is an identity matrix and $\Lambda \subset [\pp]$ is any index set satisfying $|\Lambda| \le \qq$, then, as mentioned in \cite{Shouky14}, the problem of interest is strong observability for any $\qq$-sparse identity matrix $I_\Lambda$, and the design of a UIO-based estimator for unknown $\Lambda$.}

In the first part of this paper, a vulnerability analysis is conducted.
Fundamental limitations such as attack detectability (and identifiability) conditions have been investigated in \cite{Pasqualetti13} and the attack detectability is quantified by the security index \cite{Hendrickx14}, which is the minimum number of attacks to remain undetectable.
This security index concept for a static output map is generalized to a dynamical system under sensor attacks in \cite{Chen15}. 
We have carefully explained the relationship between these fundamental limitations and the {\it redundant observability}, which is a kind of the analytical redundancy in measurements and will be formally defined in Section \ref{sec:red}.
Furthermore, equivalent conditions between them are also presented.

In the second part of this paper, we propose a resilient and robust state estimation scheme.
Compared with the existing resilient estimation algorithms in \cite{Fawzi14,Pajic14,Shouky14,Shou15ACC,Shouky16,Chong15,Jeon16}, the advantages of our scheme are as follows.
First, it does not require any additional restrictive conditions other than the redundant observability (compared with \cite{Fawzi14,Shouky14,Shouky16,Jeon16}).
Second, an observer-based algorithm makes it possible to estimate the current state, not the initial state or delay information (compared with \cite{Fawzi14,Pajic14,Shou15ACC}).
Third, the scheme is robust in the sense that a bound on estimation error is explicitly derived from system parameters (compared with \cite{Fawzi14,Pajic14,Shouky14,Shou15ACC,Shouky16,Chong15}).
Finally, the scheme requires less computational effort and less memory owing to the reduction in time and space complexity (compared with \cite{Chong15}).

The rest of the paper is organized as follows. 
Section \ref{sec:err} presents the theoretical background of static error correcting problems for a stacked vector case. 
We then present the relationship between redundant observability and security related concepts such as dynamic security index, attack detectability, and observability under attacks in Section \ref{sec:red}.
In addition, partial observers using the Kalman observability decomposition are designed and the overall resilient and robust estimation scheme is presented in Section \ref{sec:est}. 
Finally, simulation results with a three inertia system are given in Section \ref{sec:example} and we provide concluding remarks in Section \ref{sec:con}.

%%%%%%%%%%%%%%%%%%%%%%
\section{Static Error Correction for Stacked Vector} \label{sec:err}
%%%%%%%%%%%%%%%%%%%%%%
%Theoretical Basis of 	

In this section, a static error correcting algorithm is studied that will play a key role for constructing the decoder $\DD$ in the estimator $\EE$.
In particular, we solve a particular problem: 
given a matrix $\Phi \in \R^{{\np}\times \nn}$, recover an unknown vector $x \in \R^\nn$ from the known measurement $\hat z$ given by\footnote{Later on, the analysis in Section \ref{sec:red} is performed based on the measurement equation \eqref{eq:baryn} and the design in Section \ref{sec:est} is carried out based on the estimation error equation \eqref{eq:dyn_hatz}. Note that both equations are in the form of \eqref{eq:static_y}.} 
\begin{equation} \label{eq:static_y}
\hat z = \Phi x + v + e \in \R^{\np},
\end{equation}
where the $\nn$-stacked vector $\hat z \in \R^{\np}$ is corrupted by two more unknown vectors $v \in \R^{\np}$ and $e \in \R^{\np}$.
The vector $v$ represents noise and is assumed to have bounded magnitude.
The vector $e$ is called error, and it corresponds to an attack signal whose magnitude can be arbitrarily large but is assumed to be sparse.
The matrix $\Phi$ is called a {\em coding matrix}.

%%%%%%%%%%%%%%%%%%%%%%
\subsection{Error Detectability and Detection Scheme} \label{subsec:e_detec}
%%%%%%%%%%%%%%%%%%%%%%

One should be able to detect the existence of an error to reconstruct the original state vector $x$.
Thus, we start this subsection by introducing the notion of {\em error detectability} when the measurement $\hat z$ in \eqref{eq:static_y} is noise-free (i.e., $v = 0_{\np \times 1}$).

% % % Definition: defn:correct
\begin{defn} \label{defn:detect}
	A coding matrix $\Phi \in \R^{{\np}\times \nn}$ is said to be {\em ($\nn$-stacked) $\qq$-error detectable} if, for all $x, x' \in \R^\nn$ and $e \in \Sigma_\qq^\nn$ such that $\Phi x + e = \Phi x'$, it holds that $x = x'$.	
\end{defn}

Therefore, the matrix $\Phi \in \R^{{\np}\times \nn}$ is not ($\nn$-stacked) $\qq$-error detectable if and only if there are two different $x$ and $x'$ in $\R^\nn$, and $e$ in $\Sigma_\qq^\nn$ such that $\Phi x + e = \Phi x'$. 
Now, two more equivalent conditions that characterize the error detectability of a coding matrix $\Phi$ are given.

% % % Lemma
\begin{prop} \label{lem:detec_equiv}
	The following are equivalent:\\
	(i) the matrix $\Phi \in \R^{{\np}\times \nn}$ is ($\nn$-stacked) $\qq$-error detectable;\\
	(ii) for every set $\Lambda \subset [\pp]$ satisfying $|\Lambda| \ge \pp-\qq$, $\Phi_{\Lambda^\nn}$ (or, equivalently, $\Phi_{\Lambda^\nn}^\pi$) has full column rank;\\
	(iii) for any $x \in \R^\nn$ where $x \neq 0_{\nn \times 1}$, $\|\Phi x \|_{0^\nn} > \qq$. 
\end{prop}

\begin{proof}
	(i) $\Rightarrow$ (ii): Suppose that (ii) does not hold, i.e., there exists an index set $\Lambda \subset [\pp]$ with $|\Lambda| \ge \pp-\qq$ and $x \neq 0_{\nn \times 1}$ such that $\Phi_{\Lambda^\nn} x = 0_{\np \times 1}$.
	Then it follows that $\|e\|_{0^\nn} \le \qq$ where $e := -\Phi x$. 
	Thus, $\Phi x + e = \Phi 0_{\nn \times 1}$, and $\Phi$ is not $\qq$-error detectable.
	
	\noindent
	(ii) $\Rightarrow$ (iii):  Suppose, for the sake of contradiction, that there exists $x \neq 0_{\nn \times 1}$ such that $\|\Phi x \|_{0^\nn} \le \qq$. 
	Let $\Lambda$ be the complement of $\supp^\nn(\Phi x)$, i.e., $\Lambda = \left(\supp^\nn(\Phi x)\right)^c$. 
	Then it is obvious that $|\Lambda| \ge \pp-\qq$ and $\Phi_{\Lambda^\nn} x = 0_{\np \times 1}$. 
	This contradicts the full column rank condition of $\Phi_{\Lambda^\nn}$ in (ii).
	
	\noindent
	(iii) $\Rightarrow$ (i): We again prove it by contradiction. Suppose that $\Phi$ is not $\qq$-error detectable. That is, there exist $x, x' \in \R^\nn$ satisfying $x \neq x'$, and $e \in \Sigma_\qq^\nn$ such that $\Phi x + e = \Phi x'$. It follows from $x'-x\neq 0_{\nn \times 1}$ and $e\in \Sigma_{\qq}^\nn$ that $\|\Phi (x' - x)\|_{0^\nn} = \|e\|_{0^\nn} \le \qq$.
	Thus, condition (iii) does not hold.
\end{proof}

\begin{rem} \label{rem:redun_detect}
	In Proposition \ref{lem:detec_equiv}, condition (ii) relates $\qq$-error detectability to the left invertibility of $\Phi$.
	That is, $\Phi$ remains left invertible even if any ($\nn$-stacked) $\qq$ row blocks are eliminated. 
	We may call this property {\em $\qq$-redundant left invertibility}.
	On the other hand, condition (iii) establishes the link between the error detectability and the cospark of a coding matrix.
	More specifically, $\Phi$ is $\qq$-error detectable if and only if its cospark is larger than $\qq$, i.e., $\cospark^\nn(\Phi) > \qq$.
\end{rem}

The equivalence conditions in Proposition \ref{lem:detec_equiv} lead to a criterion of $\qq$-sparse error detection based on a residual signal  
\begin{equation} \label{eq:residual}
r := \hat z -  \Phi \Phi^\dagger \hat z = \left(I_{\np\times\np} -\Phi (\Phi^\top \Phi)^{-1}\Phi^\top \right) \hat z .
\end{equation}

% % % Theorem
\begin{lem} \label{thm:Detec_noiseless} 
	For the measurement $\hat z = \Phi x + e$ where $\Phi \in \R^{{\np}\times \nn}$ is ($\nn$-stacked) $\qq$-error detectable, $x \in \R^\nn$, and $e \in \Sigma_\qq^\nn$, let $r = \hat z - \Phi \Phi^{\dagger} \hat z$.
	Then $e = 0_{\np\times 1}$ if and only if $r = 0_{\np\times 1}$.
	Moreover, when $e = 0_{\np \times 1}$, the vector $x$ is recovered by $\hat x := \Phi^{\dagger} \hat z$.
\end{lem}

\begin{proof}
	Note that any non-zero $\qq$-sparse error $e$ does not lie in $\RR(\Phi)$ by Proposition \ref{lem:detec_equiv}.(iii).
	Hence, $e \not = 0_{\np \times 1}$ is equivalent to the condition that $\hat z  = \Phi x + e \notin \RR(\Phi)$.
	Since $\Phi \Phi^\dagger$ is a projection matrix and it projects $\hat z$ onto $\RR(\Phi)$, we have $\hat z \notin \RR(\Phi)$ if and only if $\hat z \not =  \Phi \Phi^\dagger \hat z$.
	This completes the proof.
\end{proof}

Inspired by the error detection scheme for the noiseless case of Lemma \ref{thm:Detec_noiseless}, let us now consider a scheme for the case when the bounded noise $v\in\R^\np$ corrupts the measurements. 
For this, let
\begin{align*}
\begin{split}
\rho_{\pp,\qq}(\Phi) &:= \min \left\{ \sigma_{\min}\left(\Phi_{\Lambda^\nn}\right) : \Lambda \subset [\pp],~|\Lambda| = \pp-\qq \right\},
\end{split}\\
\begin{split}
\eta_{\pp,\qq} (\Phi) &:= \max \Big\{ \big\| \Phi_{\Gamma_\ii^\nn} \left(\Phi_{\Lambda^\nn}\right)^{\dagger} \big\|_2 : \ii \in [\pp]\setminus\Lambda, \\
&\qquad \qquad \qquad \qquad \qquad \quad~ \Lambda \subset [\pp],~|\Lambda| = \pp-\qq \Big\}, 
\end{split}\\
\begin{split}
\kappa_{\pp,\qq}^d (\Phi) &:= (\sqrt{\pp}+1)\sqrt{\pp-\qq}/{\rho_{\pp,\qq}(\Phi)},
\end{split}\\
\begin{split}
\kappa_{\pp,\qq}^e (\Phi) &:= \left(\eta_{\pp,\qq}(\Phi) \sqrt{\pp-\qq} +1 \right)(\sqrt{\pp}+1).
\end{split}
\end{align*}
Then, the following theorem says that one can ``practically'' detect the $\qq$-sparse error in the noisy situation with the residual $r$ given in \eqref{eq:residual}.

% % % Theorem
\begin{thm} \label{thm:detec_thmv}
	For the measurement $\hat z = \Phi x + v + e$ where $\Phi \in \R^{{\np}\times \nn}$ is ($\nn$-stacked) $\qq$-error detectable, $x \in \R^\nn$, $e \in \Sigma_\qq^\nn$, and $v \in \R^\np$ satisfying $\|v_\ii^\nn\|_2 \le v_{\max}$, $^\forall \ii \in [\pp]$, let $\hat x = \Phi^{\dagger} \hat z$ and $r = \hat z - \Phi \hat x$.
	Then: 
	
	\noindent
	(i) $e \neq 0_{\nn\pp \times 1}$ if
	$$\|r_\ii^\nn\|_2 = \|\hat z_\ii^\nn-\Phi_{\Gamma_\ii^\nn}^\pi \hat x \|_2  > \sqrt{\pp}\:v_{\max} \quad \text{for some $\ii \in [\pp]$;}$$
	\noindent
	(ii) $\|e_\ii^\nn\|_2 \le \kappa_{\pp,\qq}^e (\Phi) v_{\max}$, $^\forall \ii \in [\pp]$, if
	$$\|r_\ii^\nn\|_2 = \|\hat z_\ii^\nn-\Phi_{\Gamma_\ii^\nn}^\pi \hat x \|_2  \le \sqrt{\pp}\:v_{\max} \quad \text{for all $\ii \in [\pp]$.}$$
	In the case of (ii), $\|\hat x - x\|_2 \le \kappa_{\pp,\qq}^d (\Phi) v_{\max}$.
\end{thm}

\begin{proof}		
(i): This can be proved by contraposition.
If $e = 0_{\nn\pp \times 1}$, then we have, for all $\ii \in [\pp]$,
\begin{align*}
\begin{split}
&\|\hat z_\ii^\nn-\Phi_{\Gamma_\ii^\nn}^\pi \hat x \|_2  = \|\Phi_{\Gamma_\ii^\nn}^\pi x + v_\ii^\nn -\Phi_{\Gamma_\ii^\nn}^\pi (\Phi^{\dagger} (\Phi x + v)) \|_2 \\
&= \| v_\ii^\nn -\Phi_{\Gamma_\ii^\nn}^\pi (\Phi^{\dagger} v) \|_2\le \| (I_{\np\times\np} -\Phi \Phi^{\dagger}) v \|_2 \le \sqrt{\pp}\:v_{\max}, \\
\end{split}
\end{align*}
which follows from the fact that $\| I_{\np \times \np } - \Phi \Phi^{\dagger} \|_2 \le 1$.

\noindent
(ii): Let $\Lambda$ be a subset of $\left(\supp^\nn(e)\right)^c$ satisfying $|\Lambda| = \pp-\qq$. 
Since $\sqrt{\pp}\:v_{\max} \ge \|\hat z_\ii^\nn-\Phi_{\Gamma_\ii^\nn}^\pi \hat x \|_2$ for all $\ii\in\Lambda$ from the assumption, we have
\begin{align*}
\begin{split}
&\sqrt{\pp(\pp-\qq)}\:v_{\max} \ge \|\hat z_{\Lambda^\nn}-\Phi_{\Lambda^\nn} \hat x \|_2  = \|\Phi_{\Lambda^\nn} x + v_{\Lambda^\nn} -\Phi_{\Lambda^\nn} \hat x \|_2 \\
&= \| \Phi_{\Lambda^\nn} (x - \hat x) +v_{\Lambda^\nn} \|_2 
\ge \| \Phi_{\Lambda^\nn} (x - \hat x)\|_2 - \|v_{\Lambda^\nn}\|_2,
\end{split}
\end{align*}
which leads to the result that 
$$\| \Phi_{\Lambda^\nn} (x - \hat x)\|_2 \le (\sqrt{\pp}+1)\sqrt{\pp-\qq}\: v_{\max}.$$
Therefore, it is obtained that
$$\|\hat x - x\|_2 \le (\sqrt{\pp}+1)\sqrt{\pp-\qq}\:v_{\max}/{\rho_{\pp,\qq}(\Phi)}  = \kappa_{\pp,\qq}^d (\Phi) v_{\max}.$$
Now, for any $\ii\in\Lambda^c$, it follows again from the assumption that
\begin{align*}
\begin{split}
\sqrt{\pp}\: v_{\max} & \ge \|\hat z_\ii^\nn-\Phi_{\Gamma_\ii^\nn}^\pi \hat x \|_2  = \|\Phi_{\Gamma_\ii^\nn}^\pi x + v_\ii^\nn + e_\ii^\nn -\Phi_{\Gamma_\ii^\nn}^\pi \hat x \|_2 \\
&= \| \Phi_{\Gamma_\ii^\nn}^\pi (x - \hat x) +v_\ii^\nn + e_\ii^\nn\|_2 \\
&\ge -\| \Phi_{\Gamma_\ii^\nn}^\pi (x - \hat x)\|_2 - v_{\max} + \|e_\ii^\nn\|_2 .
\end{split}
\end{align*}
Hence, 
\begin{align*}
\begin{split}
\|e_\ii^\nn\|_2 &\le \| \Phi_{\Gamma_\ii^\nn}^\pi (x - \hat x)\|_2 + \sqrt{\pp}\:v_{\max} + v_{\max}\\
&= \| \Phi_{\Gamma_\ii^\nn} \Phi_{\Lambda^\nn}^{\dagger} \Phi_{\Lambda^\nn}(x - \hat x)\|_2 + (\sqrt{\pp}+1) v_{\max}\\
&\le\eta_{\pp,\qq}(\Phi) (\sqrt{\pp}+1) \sqrt{\pp-\qq}\:v_{\max} + (\sqrt{\pp}+1) v_{\max}\\
&= \kappa_{\pp,\qq}^e (\Phi) v_{\max}, \quad ^\forall \ii\in\Lambda^c.
\end{split}
\end{align*}
Since $\|e_\ii^\nn\|_2 = 0 $ for all $\ii\in\Lambda$, this completes the proof.
\end{proof}

In fact, when the magnitude of $e$ is small, one cannot differentiate between the noise $v$ and the error $e$.
Theorem~\ref{thm:detec_thmv}.(ii) reflects this fact and guarantees that the estimation error is small and $\hat x$ approximately estimates $x$.

%%%%%%%%%%%%%%%%%%%%%%
\subsection{Error Correctability and Reconstruction Scheme}
%%%%%%%%%%%%%%%%%%%%%%

In the noiseless case, the following notion of {\em error correctability} is introduced and characterized in this subsection.

% % % Definition: defn:correct
\begin{defn} \label{defn:correct}
	A coding matrix $\Phi \in \R^{{\np}\times \nn}$ is said to be {\em ($\nn$-stacked) $\qq$-error correctable} if, for all $x_1, x_2 \in \R^\nn$ and $e_1, e_2 \in \Sigma_\qq^\nn$ such that $\Phi x_1 + e_1 = \Phi x_2 + e_2$, it holds that $x_1 = x_2$.
\end{defn}

Now, one can easily obtain the following equivalence between the error correctability and the error detectability.
The following proposition implies that one can detect twice the number of errors that can be corrected and reconstructed.

% % % Lemma
\begin{prop} \label{lem:corr_equiv}
	The following are equivalent:\\
	(i) the matrix $\Phi \in \R^{{\np}\times \nn}$ is ($\nn$-stacked) $\qq$-error correctable;\\
	(ii) the matrix $\Phi \in \R^{{\np}\times \nn}$ is ($\nn$-stacked) $2\qq$-error detectable.
\end{prop}

\begin{proof}
	(i) $\Rightarrow$ (ii): Assume that $x, x' \in \R^\nn$ and $e \in \Sigma_{2\qq}^\nn$ satisfying $\Phi x + e = \Phi x'$ are given.
	Let $e_1$ and $e_2$ be such that $e = e_1 - e_2$ where  $e_1, e_2 \in \Sigma_\qq^\nn$.
	Thus, we have $\Phi x + e_1 = \Phi x' + e_2$. 
	Since $\Phi \in \R^{{\np}\times \nn}$ is $\qq$-error correctable, it follows that $x =x'$. 
	
	\noindent
	(ii) $\Rightarrow$ (i):  Assume that $x_1, x_2 \in \R^\nn$ and $e_1, e_2 \in \Sigma_\qq^\nn$ satisfying $\Phi x_1 + e_1 = \Phi x_2 + e_2$ are given.
	Then, we have $\Phi x_1 + e = \Phi x_2$ where $e = e_1 - e_2 \in \Sigma_{2\qq}^\nn$. 
	Since $\Phi \in \R^{{\np}\times \nn}$ is $2\qq$-error detectable, it follows that $x_1 = x_2$. 
\end{proof}

Based on the notion of $\qq$-error correctability, we discuss the problem of constructing a decoder that can actually correct ($\nn$-stacked) $\qq$ errors and recover the original state $x$  when $v=0_{\np \times 1}$ in \eqref{eq:static_y}. 
That is, we find a map $\DD : \R^{\np} \rightarrow \R^\nn$ such that $\DD(\hat z)=x$ where $\hat z = \Phi x + e \in \R^\np$ and $e \in \Sigma_\qq^\nn$.
This is basically achieved through $\ell_0$ minimization \cite[Section 3]{Guruswami08}. 
Here we claim that searching over a finite set is enough to solve the minimization problem.

% % % Theorem
\begin{thm} \label{thm:finite_opt}
	For the measurement $\hat z = \Phi x + e \in \R^{\np}$ with ($\nn$-stacked) $\qq$-error correctable $\Phi \in \R^{{\np}\times \nn}$, $x \in \R^\nn$, and $e \in \Sigma_\qq^\nn$, it follows that 
	\begin{align}
	x &= \argmin_{\chi\in\mathcal{F}_{\pp,\rr}(\hat z)} \|\hat z-\Phi \chi\|_{0^\nn} \label{eq:nonoise}\\
	& = \argmin_{\chi\in\mathcal{F}_{\pp,\rr}(\hat z)} \big| \big\{ \ii \in [\pp] : \|\hat z_\ii^\nn-\Phi_{\Gamma_\ii^\nn}^\pi \chi \|_2  > 0 \big\} \big|, \tag{\ref{eq:nonoise}$'$}\label{eq:nonoise'}
	\end{align}
	where 
	\begin{align*}
	\begin{split}
	\mathcal{F}_{\pp,\rr}(\hat z) := \big\{ \left(\Phi_{\Lambda^\nn}\right)^{\dagger} \hat z_{\Lambda^\nn} \in \R^\nn : \Lambda \subset [\pp],~ |\Lambda| = \pp-\rr \big\}
	\end{split}
	\end{align*}
	and $\rr$ is any integer satisfying $\qq \le \rr \le 2\qq$.
\end{thm}

\begin{proof}
	We first show that the vector $x$ belongs to $\mathcal{F}_{\pp,\rr}(\hat z)$. 
	Pick any subset $ \Lambda \subset \left(\supp^\nn(e)\right)^c$ satisfying $| \Lambda| = \pp-\rr$. 
	Because $\Phi_{\Lambda^\nn} $ has full column rank by Propositions \ref{lem:detec_equiv}.(ii) and \ref{lem:corr_equiv}, it follows that $\chi = \left( \Phi_{\Lambda^\nn} \right)^{\dagger} \hat z_{\Lambda^\nn} =  \left( \Phi_{\Lambda^\nn} \right)^{\dagger} \Phi_{\Lambda^\nn} x = x$.
	Hence, $x \in \mathcal{F}_{\pp,\rr}(\hat z)$.
	Now, it suffices to show that $x$ is a minimizer of $\|\hat z-\Phi \chi\|_{0^\nn}$. 
	Suppose, for the sake of contradiction, that there exists $x' \not = x$ in $\mathcal{F}_{\pp,\rr}(\hat z)$ that minimizes $\|\hat z-\Phi \chi\|_{0^\nn}$, then, with $e' := \hat z - \Phi x'$, we have that $\hat z = \Phi x' + e' = \Phi x + e$ and $\|e'\|_{0^\nn} \le \|e\|_{0^\nn} \le \qq$ because $e'$ is a minimal solution. 
	This contradicts the assumption that $\Phi$ is $\qq$-error correctable. 
\end{proof}

This theorem claims that it is enough to search over the finite set $\mathcal{F}_{\pp,\rr}(\hat z)$, not the whole space $\R^\nn$, to solve \eqref{eq:nonoise}. 
Keeping in mind the fact that $|\mathcal{F}_{\pp,\rr}(\hat z)| \le \binom{\pp}{\pp-\rr} = \binom{\pp}{\rr}$, one can choose any integer $\rr$ between $\qq$ and $2\qq$ to minimize $\binom{\pp}{\rr}$.

% % % Remark:
\begin{rem} \label{rem:nphard}
	The $\ell_0$ minimization problem over $\R^n$ is shown to be NP-hard \cite{Natarajan95}. 
	Whereas previous research efforts have been devoted to a relaxation of the problem by imposing some additional conditions (e.g., \cite{Fawzi14,Shouky14}), Theorem \ref{thm:finite_opt} actually relieves the computational complexity by reducing the search space to a {\em finite} set.
	It is a kind of combinatorial approach that tests only $\binom{\pp}{\rr} \le \pp^\rr$ (or $\binom{\pp}{\pp-\rr}\le \pp^{\pp-\rr}$) candidates with the freedom of selecting $\rr$ between $\qq$ and $2\qq$, whereas naive brute-force search algorithm without any information on error correctability has no choice but to test all $\binom{\pp}{1}+\binom{\pp}{2}+\cdots+\binom{\pp}{\pp} \approx 2^\pp$ combinations.
	In our case, the computational efforts decrease drastically by selecting $\rr=\qq$ when $\qq \ll \pp$ (or selecting $\rr=2\qq$ when $\qq \approx \pp/2$) for example. 
	Compared with other combinatorial algorithms in \cite{Pasqualetti13,Lee15,Shou15ACC}, Theorem \ref{thm:finite_opt} is more relaxed by introducing $\rr$ that can vary between $\qq$ and $2\qq$.
\end{rem}

Finally, the following lemma presents a simple criterion to verify whether a given vector $\hat x \in \R^\nn$ coincides with the original input $x$.

% % % Lemma
\begin{lem} \label{lem:suff_q}
	For the measurement $\hat z = \Phi x + e \in \R^{\np}$ with ($\nn$-stacked) $\qq$-error correctable $\Phi \in \R^{{\np}\times \nn}$, $x \in \R^\nn$, and $e \in \Sigma_\qq^\nn$, 
	$$\|\hat z-\Phi \hat x\|_{0^\nn} \le \qq \quad \text{if and only if} \quad x = \hat x.$$	
\end{lem}

\begin{proof}
	(if): This is trivial because $\|\hat z-\Phi \hat x\|_{0^\nn} = \|e\|_{0^\nn} \le \qq$.\\
	(only if): Define $\hat e := \hat z-\Phi \hat x$, then $\hat z=\Phi \hat x + \hat e = \Phi x + e$ where  $e, \hat e \in \Sigma_\qq^\nn$.
	Since $\Phi$ is $\qq$-error correctable, it follows from Definition \ref{defn:correct} that $x = \hat x$.
\end{proof}

Now, bounded noise $v \in \R^\np$ satisfying $\|v_\ii^\nn\|_2 \le v_{\max}$ for all $\ii \in [\pp]$ is taken into account and a state recovery scheme estimating $x$ is presented.
More precisely, we show that any solution $(\chi^*,\varepsilon^*)$ to the following relaxed $\ell_0$ minimization problem yields an approximation of $x$ as $\hat x = \chi^*$:
\begin{align} \label{eq:finite_opt_v}
\begin{split}
&\min_{\chi\in\mathcal{F}_{\pp,\rr}(\hat z),~\varepsilon\in\R^\np} \|\varepsilon\|_{0^\nn} \\
&~\!{\rm subject~to~} \|\hat z_\ii^\nn-\Phi_{\Gamma_\ii^\nn}^\pi \chi - \varepsilon_\ii^\nn\|_2 \le v_{\max}', ~~ {}^\forall \ii \in [\pp],
\end{split}
\end{align}
where $\rr$ is any integer satisfying $\qq \le \rr \le 2\qq$ and
\begin{align*} \label{eq:v'max}
\begin{split}
v_{\max}' &:= \vartheta_{\pp,\qq,\rr} (\Phi) v_{\max}\\
&:= \max \left\{ \eta_{\pp,\qq,\rr}'(\Phi)\sqrt{\pp-\rr} +1, \sqrt{\pp-\rr} \right\} v_{\max},
\\
\eta_{\pp,\qq,\rr}' (\Phi) &:= \max_{\substack{\Lambda \subset [\pp] \\ |\Lambda| = \pp-\qq}} ~\! \min_{\substack{\bar\Lambda \subset \Lambda \\ |\bar\Lambda| = \pp-\rr}} ~\! \max_{\ii \in \Lambda\setminus\bar\Lambda} \big\| \Phi_{\Gamma_\ii^\nn} \left(\Phi_{\bar\Lambda^\nn}\right)^{\dagger} \big\|_2. 
\end{split}
\end{align*}
The above optimization problem is not easily implementable because the variable $\varepsilon$ is searched over $\R^\np$ under constraints.
Hence, we present another optimization problem, which may be considered as a relaxation of \eqref{eq:nonoise'}:
\begin{equation} \tag{\ref{eq:finite_opt_v}$'$}\label{eq:finite_opt_v'}
\hat x = \argmin_{\chi\in\mathcal{F}_{\pp,\rr}(\hat z)} \big| \big\{ \ii \in [\pp] : \|\hat z_\ii^\nn-\Phi_{\Gamma_\ii^\nn}^\pi \chi \|_2  > v_{\max}'\big\} \big| .
\end{equation}
Whereas the problem \eqref{eq:finite_opt_v} or \eqref{eq:finite_opt_v'} need not have a unique solution, the following theorem shows equivalence between \eqref{eq:finite_opt_v} and \eqref{eq:finite_opt_v'}, and presents an upper bound of $\|\hat x - x\|_2$ for any solution $\hat x$ of \eqref{eq:finite_opt_v} or \eqref{eq:finite_opt_v'}.

% % % Theorem
\begin{thm} \label{thm:finite_thmv}
	For the measurement $\hat z = \Phi x + e + v \in \R^{\np}$ with ($\nn$-stacked) $\qq$-error correctable $\Phi \in \R^{{\np}\times \nn}$, $x \in \R^\nn$, $e \in \Sigma_\qq^\nn$, and $v \in \R^\np$ such that $\|v_\ii^\nn\|_2 \le v_{\max}$, $^\forall \ii \in [\pp]$, the following hold:
	
	\noindent
	(i) two optimization problems \eqref{eq:finite_opt_v} and \eqref{eq:finite_opt_v'} are equivalent (that is, a solution $\hat x$ to \eqref{eq:finite_opt_v} is also a solution to \eqref{eq:finite_opt_v'} and vice versa);
	
	\noindent
	(ii) for any solution $\hat x$, $\|\hat x - x\|_2 \le \kappa_{\pp,\qq,\rr}^c (\Phi)~v_{\max}$ where
	$$\kappa_{\pp,\qq,\rr}^c (\Phi) := (\vartheta_{\pp,\qq,\rr} (\Phi) +1) \sqrt{\pp-2\qq}/{\rho_{\pp,2\qq}(\Phi)}.$$
\end{thm}

\begin{proof}
(i): Let $\hat x = \chi^*$ and $\hat e = \varepsilon^*$ be any solution to \eqref{eq:finite_opt_v}, and let $\hat v := \hat z - \Phi \hat x - \hat e$.
Then, for any $\ii \in [\pp]$, it automatically holds that $\|\hat v_\ii^\nn\|_2 \le v_{\max}'$ by the constraint in \eqref{eq:finite_opt_v}.
Similarly, let $\hat x'$ be the solution to \eqref{eq:finite_opt_v'}.
Define $\hat e'{}_\jj^\nn := \hat z_\jj^\nn-\Phi_{\Gamma_\jj^\nn}^\pi \hat x'$ and $\hat v'{}_\jj^\nn := 0_{\nn \times 1}$ for $\jj \in \big\{ \ii \in [\pp]: \|\hat z_\ii^\nn-\Phi_{\Gamma_\ii^\nn}^\pi \hat x' \|_2  > v_{\max}'\big\}$, and define $\hat e'{}_\jj^\nn :=0_{\nn \times 1}$ and $\hat v'{}_\jj^\nn := \hat z_\jj^\nn-\Phi_{\Gamma_\jj^\nn}^\pi \hat x'$ for $\jj \in \big\{\ii \in [\pp] : \|\hat z_\ii^\nn-\Phi_{\Gamma_\ii^\nn}^\pi \hat x' \|_2  \le v_{\max}'\big\}$.
Then, $(\chi,\varepsilon) = (\hat x', \hat e')$ satisfies the constraint in \eqref{eq:finite_opt_v}.

We claim that $\hat x$ with $\hat e$, the solution of \eqref{eq:finite_opt_v}, is also a solution of \eqref{eq:finite_opt_v'} and vice versa.
Indeed, directly from the above definition of $\hat e'$, it is obtained that
\begin{align} \label{eq:ineq1}
\big| \big\{ \ii \in [\pp] : \|\hat z_\ii^\nn-\Phi_{\Gamma_\ii^\nn}^\pi \hat x' \|_2  > v_{\max}'\big\} \big| = \|\hat e'\|_{0^\nn}.
\end{align}
On the other hand, because $\|\hat z_\ii^\nn-\Phi_{\Gamma_\ii^\nn}^\pi \hat x - \hat e_\ii^\nn\|_2 \le v_{\max}'$ for all $\ii \in [\pp]$, it follows that $\|\hat z_\jj^\nn-\Phi_{\Gamma_\jj^\nn}^\pi \hat x \|_2 \le v_{\max}'$ for any $\jj \in \big(\supp^\nn(\hat e)\big)^c$. 
Thus, we have
\begin{align} \label{eq:ineq2}
\big| \big\{ \ii \in [\pp] : \|\hat z_\ii^\nn-\Phi_{\Gamma_\ii^\nn}^\pi \hat x \|_2  > v_{\max}'\big\} \big| \le  \|\hat e\|_{0^\nn}.
\end{align}
Since $\hat e$ is the minimal solution of \eqref{eq:finite_opt_v}, it holds that
\begin{align} \label{eq:ineq3}
\|\hat e\|_{0^\nn} \le \|\hat e'\|_{0^\nn}.
\end{align}
Finally, because $\hat x'$ is the solution of \eqref{eq:finite_opt_v'}, it follows that
\begin{align} \label{eq:ineq4}
\begin{split}
&\big| \big\{ \ii \in [\pp] : \|\hat z_\ii^\nn-\Phi_{\Gamma_\ii^\nn}^\pi \hat x' \|_2 > v_{\max}'\big\} \big| \\
&\qquad\qquad \le \big| \big\{ \ii \in [\pp] : \|\hat z_\ii^\nn-\Phi_{\Gamma_\ii^\nn}^\pi \hat x \|_2  > v_{\max}'\big\} \big|.
\end{split}
\end{align}
Combining \eqref{eq:ineq1}, \eqref{eq:ineq2}, \eqref{eq:ineq3}, and \eqref{eq:ineq4} together results in
\begin{align*} 
\begin{split}
&\|\hat e\|_{0^\nn} =\big| \big\{ \ii \in [\pp] : \|\hat z_\ii^\nn-\Phi_{\Gamma_\ii^\nn}^\pi \hat x \|_2 > v_{\max}'\big\} \big| \\
&= \|\hat e'\|_{0^\nn} =\big| \big\{ \ii \in [\pp] : \|\hat z_\ii^\nn-\Phi_{\Gamma_\ii^\nn}^\pi \hat x' \|_2  > v_{\max}'\big\} \big|.
\end{split}
\end{align*}
Consequently, $\hat x$ is a solution of \eqref{eq:finite_opt_v'} and $\hat x'$ is a solution of \eqref{eq:finite_opt_v}. 
This concludes the claim.

\noindent
(ii): Let $(\hat x,\hat e)$ be a solution $(\chi^*,\varepsilon^*)$ of \eqref{eq:finite_opt_v}.
Then, we first show that $\|\hat e\|_{0^\nn} \le \qq$.
Let $\Lambda$ be a subset of $\left(\supp^\nn(e)\right)^c$ satisfying $|\Lambda| = \pp-\qq$.
Then, there always exists a subset $\bar\Lambda \subset \Lambda$ such that $|\bar\Lambda| = \pp-\rr$ and $\displaystyle\max_{\ii \in \Lambda\setminus\bar\Lambda} \big\| \Phi_{\Gamma_\ii^\nn} \left(\Phi_{\bar\Lambda^\nn}\right)^{\dagger} \big\|_2 \le \eta_{\pp,\qq,\rr}'(\Phi)$.
Let $\bar x := \left(\Phi_{\bar\Lambda^\nn}\right)^{\dagger} \hat z_{\bar\Lambda^\nn}$, which belongs to $\mathcal{F}_{\pp,\rr}(\hat z)$.
Then it follows that $\bar x = x + \left(\Phi_{\bar\Lambda^\nn}\right)^{\dagger} v_{\bar\Lambda^\nn}$ because $\Phi_{\bar\Lambda^\nn}$ has full column rank, and thus $\left(\Phi_{\bar\Lambda^\nn}\right)^{\dagger} \Phi_{\bar\Lambda^\nn} = I_{\nn \times \nn}$.
With $\bar x$ at hand, let us define a noise vector $\bar v := \hat z_{\Lambda^\nn} - \Phi_{\Lambda^\nn} \bar x \in \R^\np$ and an error vector $\bar e := \hat z - \Phi \bar x - \bar v$. 
Here, the vector $\bar v$ can be decomposed as
\begin{align*}
\begin{split}
\bar v &= \Phi_{\Lambda^\nn} x + v_{\Lambda^\nn} - \Phi_{\Lambda^\nn} (x + \left(\Phi_{\bar\Lambda^\nn}\right)^{\dagger} v_{\bar\Lambda^\nn})\\
&= v_{(\Lambda\setminus\bar\Lambda)^\nn} + v_{\bar\Lambda^\nn} - ( \Phi_{(\Lambda\setminus\bar\Lambda)^\nn} + \Phi_{\bar\Lambda^\nn}) \left(\Phi_{\bar\Lambda^\nn}\right)^{\dagger} v_{\bar\Lambda^\nn} \\
&= v_{(\Lambda\setminus\bar\Lambda)^\nn} - \Phi_{(\Lambda\setminus\bar\Lambda)^\nn} \left(\Phi_{\bar\Lambda^\nn}\right)^{\dagger} v_{\bar\Lambda^\nn} + (I_{\np \times \np }\! - \Phi_{\bar\Lambda^\nn} (\Phi_{\bar\Lambda^\nn})^{\dagger}) v_{\bar\Lambda^\nn}, 
\end{split}
\end{align*} 
and thus it follows that
\begin{align*}
\begin{split}
& \|\hat z_\ii^\nn-\Phi_{\Gamma_\ii^\nn}^\pi \bar x - \bar e_\ii^\nn\|_2 = \|\bar v_\ii^\nn\|_2 \\
& \le \max \left\{\eta_{\pp,\qq,\rr}'(\Phi)\sqrt{\pp-\rr}+1, \sqrt{\pp-\rr} \right\} v_{\max} = v_{\max}', {}^\forall \ii \in [\pp] ,
\end{split}
\end{align*}
in which, we use the fact that $\| I_{\np \times \np } - \Phi_{\bar\Lambda^\nn} (\Phi_{\bar\Lambda^\nn})^{\dagger} \|_2 \le 1$ and $ \|v_{\bar\Lambda^\nn} \|_2 \le \sqrt{\pp-\rr}\:v_{\max}$. 
Therefore, it is clear that $\bar x$ and $\bar e$ satisfy the constraint in \eqref{eq:finite_opt_v}, i.e., $\| \hat z_\ii^\nn - \Phi_{\Gamma_\ii^\nn}^\pi \bar x - \bar e_\ii^\nn \|_2 \le v_{\max}'$ for all $\ii \in [\pp]$.
Moreover, from the construction, $\| \bar e\|_{0^\nn} \le \qq$.
Finally, noting that $\hat e$ is the minimal solution of \eqref{eq:finite_opt_v}, we have that $\|\hat e\|_{0^\nn} \le \| \bar e\|_{0^\nn} \le \qq$.

Now, the solution $(\hat x,\hat e)$ of \eqref{eq:finite_opt_v} yields the corresponding noise vector $\hat v$ as $\hat v := \hat z - \Phi \hat x - \hat e$, which satisfies $\|\hat v_\ii^\nn\|_2 \le v_{\max}'$ for all $\ii \in [\pp]$ by the constraint of \eqref{eq:finite_opt_v}. 
Therefore, we have two expressions for the measurement $\hat z = \Phi x + v + e = \Phi \hat x + \hat v + \hat e$, and we are interested in the difference $\tilde x := \hat x - x$.
Let $\tilde e:=\hat e -e$ and $\tilde v:=\hat v - v$.
Then, $\|\tilde e\|_{0^\nn} \le 2\qq$ and $\|\tilde v_\ii^\nn\|_2 \le v_{\max}'+v_{\max} = (\vartheta_{\pp,\qq,\rr} (\Phi)+1)~v_{\max}$ for all $\ii \in [\pp]$.
Let $\tilde\Lambda$ be any subset of $\left(\supp^\nn(\tilde e)\right)^c$ such that $| \tilde\Lambda| = \pp-2\qq$. 
Then, it follows from $\Phi \tilde x + \tilde e = - \tilde v$ that $\Phi_{\tilde\Lambda^\nn} \tilde x = -\tilde v_{\tilde\Lambda^\nn}$. 
Since $\Phi_{\tilde\Lambda^\nn}$ has full column rank by Proposition \ref{lem:corr_equiv}.(ii), it follows that $\tilde x = -\left(\Phi_{\tilde\Lambda^\nn} \right)^{\dagger} \tilde v_{\tilde\Lambda^\nn}$.
Therefore, one can compute the bound of $\|\tilde x\|_2$ as $\|\tilde x\|_2 \le \big\|\left(\Phi_{\tilde\Lambda^\nn} \right)^{\dagger}\big\|_2  \|\tilde v_{\tilde\Lambda^\nn}\|_2 \le (\vartheta_{\pp,\qq,\rr} (\Phi) +1)\sqrt{\pp-2\qq} \: v_{\max}/{\rho_{\pp,2\qq}(\Phi)} = \kappa_{\pp,\qq,\rr}^c (\Phi) v_{\max} $.
\end{proof}

As in Lemma \ref{lem:suff_q}, a simple criterion to check whether a given vector $\hat x \in \R^\nn$ is close to the original $x$ with noisy measurements, is also derived in the following theorem.

% % % Lemma
\begin{thm} \label{thm:suff_qv}
	For the measurement $\hat z = \Phi x + e + v \in \R^{\np}$ with ($\nn$-stacked) $\qq$-error correctable $\Phi \in \R^{{\np}\times \nn}$, $x \in \R^\nn$, $e \in \Sigma_\qq^\nn$, and $v \in \R^\np$ such that $\|v_\ii^\nn\|_2 \le v_{\max}$, $^\forall \ii \in [\pp]$, the following hold:
	
	\noindent
	(i) $\|\hat x - x \|_2 \le \kappa_{\pp,\qq,\rr}^c (\Phi)~v_{\max}$ if $\hat x$ satisfies 
	$$\big| \big\{ \ii \in [\pp] : \|\hat z_\ii^\nn-\Phi_{\Gamma_\ii^\nn}^\pi \hat x \|_2  > v_{\max}'\big\} \big| \le \qq;$$
	
	\noindent
	(ii) $\|\hat x - x \|_2 > \kappa_{\pp,\qq,\rr}^{c'} (\Phi)~v_{\max}$ if $\hat x$ satisfies 
	$$\big| \big\{ \ii \in [\pp] : \|\hat z_\ii^\nn-\Phi_{\Gamma_\ii^\nn}^\pi \hat x \|_2  > v_{\max}'\big\} \big| > \qq,$$
	where $\displaystyle\kappa_{\pp,\qq,\rr}^{c'} (\Phi) := (\vartheta_{\pp,\qq,\rr} (\Phi) - 1)/\max_{\ii \in [\pp]} \| \Phi_{\Gamma_\ii^\nn} \|_2.$
\end{thm}

\begin{proof}
(i): With a given $\hat x$, construct the error vector $\hat e$ and the noise vector $\hat v$ as follows.
For $\jj \in \big\{ \ii \in [\pp]: \|\hat z_\ii^\nn-\Phi_{\Gamma_\ii^\nn}^\pi \hat x \|_2  > v_{\max}'\big\}$, define $\hat e{}_\jj^\nn := \hat z_\jj^\nn-\Phi_{\Gamma_\jj^\nn}^\pi \hat x$ and $\hat v{}_\jj^\nn := 0_{\nn \times 1}$.
For $\jj \in \big\{\ii \in [\pp] : \|\hat z_\ii^\nn-\Phi_{\Gamma_\ii^\nn}^\pi \hat x \|_2  \le v_{\max}'\big\}$, define $\hat e{}_\jj^\nn :=0_{\nn \times 1}$ and $\hat v{}_\jj^\nn := \hat z_\jj^\nn-\Phi_{\Gamma_\jj^\nn}^\pi \hat x$.
Then, we have two expressions for the measurement $\hat z = \Phi x + v + e = \Phi \hat x + \hat v + \hat e$, in which $\|\hat v_\ii^\nn\|_2 \le v_{\max}'$ for all $\ii \in [\pp]$ and $\|\hat e \|_{0^\nn} \le \qq$.
Therefore, the same argument in the proof of Theorem \ref{thm:finite_thmv}.(ii) applies and concludes the claim.

\noindent
(ii): This can be shown by contradiction.
Let $\|\hat x - x \|_2 \le \kappa_{\pp,\qq,\rr}^{c'} (\Phi)\: v_{\max}$.
Then, for $\ii \in \left(\supp^\nn(e)\right)^c$,
\begin{align*}
\begin{split}
\|\hat z_\ii^\nn-\Phi_{\Gamma_\ii^\nn}^\pi \hat x \|_2 &= \|\Phi_{\Gamma_\ii^\nn}^\pi (x - \hat x) + v_\ii^\nn \|_2 \\
& \le \Big( \max_{\ii \in [\pp]} \| \Phi_{\Gamma_\ii^\nn} \|_2  \Big) \kappa_{\pp,\qq,\rr}^{c'} (\Phi) v_{\max} + v_{\max} \\
& = \vartheta_{\pp,\qq,\rr} (\Phi) v_{\max} = v_{\max}'.
\end{split}
\end{align*}
Therefore, we have $\big| \big\{ \ii \in [\pp] : \|\hat z_\ii^\nn-\Phi_{\Gamma_\ii^\nn}^\pi \hat x \|_2  > v_{\max}' \big\} \big| \le \qq$ because $\|e\|_{0^\nn} \le \qq$.
\end{proof}

%%%%%%%%%%%%%%%%%%%%%%
\section{Characterization of Redundant Observability} \label{sec:red}
%%%%%%%%%%%%%%%%%%%%%%

In this section, we introduce the {\em redundant observability} and relate that concept to the {\em dynamic security index}, {\em attack detectability}, and {\em observability under sensor attacks}.
It will soon be revealed that an observability matrix behaves in the same way as a coding matrix as examined in the previous section, and hence its properties determine resilience of control systems under sensor attacks. 

%%%%%%%%%%%%%%%%%%%%%%
\subsection{Redundant Observability}
%%%%%%%%%%%%%%%%%%%%%% 

From a control theoretical viewpoint, the notion of {\em redundant observability} for the system \eqref{eq:plant} is defined as follows.  

% % % Definition:
\begin{defn} \label{def:obs_redun}
	The pair $(A, C)$ or the dynamical system \eqref{eq:plant} is said to be {\em $\qq$-redundant observable} if the pair $(A, C_{\Lambda}^\pi)$ is observable for any $\Lambda \subset [\pp]$ satisfying $|\Lambda| \ge \pp-\qq$.	
\end{defn}

To characterize the redundant observability in the following proposition, we first obtain the observability matrix $G \in \R^{{\np}\times \nn}$ as follows: 
\begin{align} \label{eq:barG}
G :=\left[G_1^\top~G_2^\top~\cdots~G_\pp^\top\right]^\top,
\end{align}
where
$G_\ii :=\left[c_\ii^\top~(c_\ii A)^\top~\cdots~(c_\ii A^{\nn-1})^\top\right]^\top$
is an observability matrix of the pair $(A, c_\ii)$.

% % % Proposition
\begin{prop} \label{lem:red_obs_equiv}
	The following are equivalent:\\
	(i) the pair $(A, C)$ is $\qq$-redundant observable;\\
	(ii) the matrix $G$ is ($\nn$-stacked) $\qq$-error detectable. 
\end{prop}

\begin{proof}
	From the fact that $G_{\Lambda^\nn}^\pi$ is the observability matrix of the pair $(A, C_{\Lambda}^\pi)$, the pair $(A, C)$ is $\qq$-redundant observable
	if and only if $G_{\Lambda^\nn}^\pi$ has full column rank for any $\Lambda \subset [\pp]$ satisfying $|\Lambda| \ge \pp-\qq$.
	Thus, the result directly follows from Proposition \ref{lem:detec_equiv}.
\end{proof}

%%%%%%%%%%%%%%%%%%%%%%
\subsection{Attack Detectability and Dynamic Security Index}
%%%%%%%%%%%%%%%%%%%%%%

Assume tentatively that there is no control input, disturbance, nor noise in the system \eqref{eq:plant} so that we can focus on the attack signal only.
Then, the output measurements for a finite time period are collected and the stacked output sequence is computed as
\begin{align} \label{eq:baryn}
\begin{split}
\bar y^{[0:\nn-1]}:=\begin{bmatrix}
\bar y_1^{[0:\nn-1]}  \\
\bar y_2^{[0:\nn-1]}  \\
\vdots \\
\bar y_\pp^{[0:\nn-1]}  \\
\end{bmatrix} &=
\begin{bmatrix}
G_1\\
G_2\\
\vdots \\
G_\pp\\
\end{bmatrix} x(0) +
\begin{bmatrix}
a_1^{[0:\nn-1]}  \\
a_2^{[0:\nn-1]}  \\
\vdots \\
a_\pp^{[0:\nn-1]} \\
\end{bmatrix}\\
&= G x(0) + a^{[0:\nn-1]},
\end{split}
\end{align}
where $\bar y_\ii^{[0:\nn-1]}:=[\bar y_\ii(0)~\bar y_\ii(1)~\cdots~\bar y_\ii(\nn-1)]^\top$ and $a_\ii^{[0:\nn-1]}:=[a_\ii(0)~a_\ii(1)~\cdots~a_\ii(\nn-1)]^\top$.
Noting that the situation is exactly the same as the noiseless case in Section \ref{subsec:e_detec} and $a^{[0:\nn-1]}$ is ($\nn$-stacked) $\qq$-sparse by Assumption \ref{ass:sparse}, we can introduce the notion of {\em $\qq$-attack detectability} of the system \eqref{eq:plant} as follows.

% % % Definition: defn:correct
\begin{defn} \label{defn:a_detect}
	The pair $(A, C)$ or the dynamical system \eqref{eq:plant} without disturbances/noises is said to be {\em $\qq$-attack detectable} if, for all $x(0), x'(0) \in \R^\nn$ and $a^{[0:\nn-1]} \in \Sigma_\qq^\nn$ such that $G x(0) + a^{[0:\nn-1]} = G x'(0)$, it holds that  $x(0) = x'(0)$.
\end{defn}

Furthermore, a direct comparison between Definitions \ref{defn:detect} and \ref{defn:a_detect} simply leads to the following proposition. 
% % % Proposition
\begin{prop} \label{lem:a_detec_equiv}
	The following are equivalent:\\
	(i) the pair $(A, C)$ is $\qq$-attack detectable;\\
	(ii) the matrix $G$ is ($\nn$-stacked) $\qq$-error detectable.
\end{prop}

As a tool for the vulnerability analysis of a system, the security index quantifies fundamental limitations on the attack detectability.
That is, the {\em dynamic security index} of the system \eqref{eq:plant}, $\alpha_d (A, C)$, is defined by the minimum number of sensor attacks for adversaries to remain undetectable and is computed by examining the system's strong observability in \cite{Chen15} as
\begin{align} \label{eq:secu_index2}
\alpha_d (A, C) = \min_{v\in\VV(A)} \| C v\|_0.
\end{align}
It is shown in the following proposition that the dynamic security index can also be characterized by the error detectability of the observability matrix $G$ through its cospark.

% % % Lemma
\begin{prop} \label{prop:cal_d_secu_index}
	For $\alpha_d (A, C)$ given in \eqref{eq:secu_index2}, it holds that
	\begin{align} \label{eq:secu_index1}
	\alpha_d (A, C) = \min_{x\in\R^\nn,\!~ x\neq 0_{\nn\times 1}} \| G x\|_{0^\nn} = \cospark^\nn(G).
	\end{align}
\end{prop}

\begin{proof}
	When $A v = \lambda v$, one can trivially check that
	$$ \min_{v\in\VV(A)} \| C v\|_0 = \min_{v\in\VV(A)} \| G v\|_{0^\nn}$$
	since $G_\ii v = \left[c_\ii v~\lambda c_\ii v~\cdots~\lambda^{\nn-1} c_\ii v\right]^\top$.
	Noting that
	$$\min_{x\in\C^\nn,\!~ x\neq 0_{\nn\times 1}} \| G x\|_{0^\nn} = \min_{x\in\R^\nn,\!~ x\neq 0_{\nn\times 1}} \| G x\|_{0^\nn}$$ 
	because $G\in{\R^{\np\times \nn}}$ is a real matrix, it suffices to show that
	$$ \min_{v\in\VV(A)} \| G v\|_{0^\nn} = \min_{x\in\C^\nn,\!~ x\neq 0_{\nn\times 1}} \| G x\|_{0^\nn}.$$
	Now, we claim that there exists $v^*\in\VV(A)$ such that 
	$$ \| G v^*\|_{0^\nn} = \min_{x\in\C^\nn,\!~ x\neq 0_{\nn\times 1}} \| G x\|_{0^\nn}.$$
	Let us denote the optimal value of the problem \eqref{eq:secu_index1} by
	$$\alpha^* := \min_{x\in\R^\nn,\!~ x\neq 0_{\nn\times 1}} \| G x\|_{0^\nn}.$$
	By the equivalence between Proposition \ref{lem:detec_equiv}.(ii) and (iii) with the observability matrix $G$, there exists an index set $\Lambda \subset [\pp]$ satisfying $|\Lambda| = \pp - \alpha^*$ such that the observability matrix $G_{\Lambda^\nn}^\pi$ does not have full column rank but the observability matrix $G_{(\Lambda\cup\{\ii\})^\nn}^\pi$ has full column rank for every $\ii\in\Lambda^c$.
	That is, the pair $(A, C_\Lambda^\pi)$ is not observable but the pair $(A, C_{\Lambda\cup\{\ii\}}^\pi )$ is observable for every $\ii\in\Lambda^c$. 
	Applying the Popov--Belevitch--Hautus (PBH) observability test, we conclude that there exist $\lambda^* \in \C$ and $v^* \in \VV(A)$ such that
	$$\begin{bmatrix}
	\lambda^* I_{\nn\times\nn}-A\\
	C_\Lambda^\pi \\
	\end{bmatrix} v^* 
	=\begin{bmatrix}
	0_{\nn\times 1}\\
	0_{(\pp-\alpha^*)\times 1}\\
	\end{bmatrix} {\rm~~and~~} c_\ii v^* \neq 0,~{}^\forall \ii \in \Lambda^c.$$
	The claim easily follows by verifying that $ \| G v^*\|_{0^\nn} = \alpha^*$.
\end{proof}

%%%%%%%%%%%%%%%%%%%%%%
\subsection{Observability under Sparse Sensor Attacks}
%%%%%%%%%%%%%%%%%%%%%% 

In this section, the notion of {\em observability under $\qq$-sparse sensor attacks} is introduced and an equivalent condition is directly derived from the definition, as follows.

% % % Definition:
\begin{defn} \label{def:obs_att}
	The pair $(A, C)$ or the dynamical system \eqref{eq:plant} without disturbances/noises is said to be {\em observable under $\qq$-sparse sensor attacks}	if the initial state $x(0)$ can be determined from the output $\bar y$ over a finite number of sampling steps with any sensor attack $a$ satisfying Assumption \ref{ass:sparse}.
\end{defn}

% % % Proposition
\begin{prop} \label{lem:obs_equiv}
	The following are equivalent:\\
	(i) the pair $(A, C)$ is observable under $\qq$-sparse sensor attacks;\\
	(ii) the matrix $G$ is ($\nn$-stacked) $\qq$-error correctable.
\end{prop}

\begin{proof}
Note that the output sequence $\bar y^{[0:\nn-1]}$ is given by $\bar y^{[0:\nn-1]}=G x(0) + a^{[0:\nn-1]} \in \R^\np$ in \eqref{eq:baryn} and $a^{[0:\nn-1]} \in \Sigma_\qq^\nn$ by Assumption \ref{ass:sparse}, the result directly follows from Definition \ref{defn:correct} which says that $G$ is ($\nn$-stacked) $\qq$-error correctable if and only if $x(0)$ can be reconstructed from the output measurements $\bar y^{[0:\nn-1]}$.  
\end{proof}

%%%%%%%%%%%%%%%%%%%%%%
\section{Design of Attack-Resilient Estimator} \label{sec:est}
%%%%%%%%%%%%%%%%%%%%%%

An attack-resilient state estimator $\EE$, which combines the partial observers $\OO_\ii$ and the decoder $\DD$, is designed in this section.
First, the partial observers $\OO_\ii$ are designed by applying the Kalman observability decomposition to each sensor output.
Second, the previously developed error correction technique tailored into this specific problem constitutes the decoder $\DD$ and it recovers the original state variable $x$.

%%%%%%%%%%%%%%%%%%%%%%
\subsection{Design of Partial Observers} \label{subsec:obs}
%%%%%%%%%%%%%%%%%%%%%%

With only one measurement $\bar y_\ii (k)$ of the plant \eqref{eq:plant}, a single-output system is obtained as follows:
\begin{align}[left = \PP_{\ii}: \empheqlbrace\,] 
\begin{split} \label{eq:iplant}
& x(k+1) = A x(k) + B u(k) + d(k) \\
& \bar y_{\ii}(k) = c_{\ii} x(k) + n_{\ii}(k) + a_{\ii}(k).
\end{split}
\end{align}
The observability matrix $G_\ii$ of \eqref{eq:iplant} is used to divide the $\nn$-dimensional state space into two subspaces. To derive a transformation matrix, first, let $\nu_{\ii}$ be the observability index of $(A, c_\ii)$, i.e., $\nu_{\ii} := {\sf rank}(G_\ii)$.
Then the set of the first $\nu_\ii$ rows of $G_\ii$ is linearly independent.
The null space of $G_\ii$, $\NN(G_\ii)$, which is $A$-invariant, is the unobservable subspace. 
Furthermore, the quotient space $\R^\nn/\NN(G_\ii)$ is sometimes called, with abuse of terminology, the observable subspace.
The matrices $Z_\ii \in \R^{\nn \times \nu_\ii}$ and $W_\ii \in \R^{\nn \times (\nn-\nu_\ii)}$ are selected such that their columns are orthonormal bases of $\NN(G_\ii)^\perp(=\RR(G_\ii^\top))$ and $\NN(G_\ii)$, respectively. 
Finally, by the Kalman observability decomposition, the state $x$ is decomposed into the observable sub-state $z_\ii \in \R^{\nu_\ii}$ and the unobservable sub-state $w_\ii \in \R^{\nn-\nu_\ii}$ with a similarity transformation 
\begin{align} \label{eq:sim_trans}
\begin{bmatrix}
z_\ii^\top & w_\ii^\top 
\end{bmatrix}^\top =
\begin{bmatrix}
Z_\ii & W_\ii
\end{bmatrix}^\top x.
\end{align}

Since it follows that $Z_\ii^\top A W_\ii = O_{\nu_\ii \times (\nn-\nu_\ii)}$ and $c_\ii W_\ii = 0_{1 \times (\nn-\nu_\ii)}$ from the construction of $Z_\ii$ and $W_\ii$, the change of variable \eqref{eq:sim_trans} leads the original single-output system \eqref{eq:iplant} to the decomposed form of
\begin{align}[left = \PP'_{\ii}:\! \empheqlbrace\,]
\begin{split} \label{eq:decomp}
& \begin{bmatrix}
z_\ii (k+1) \\
w_\ii (k+1) \\
\end{bmatrix} \!=\!
\begin{bmatrix}
Z_\ii^\top A Z_\ii & O_{\nu_\ii \times (\nn-\nu_\ii)}\\
W_\ii^\top A Z_\ii & W_\ii^\top A W_\ii\\
\end{bmatrix}
\begin{bmatrix}
z_\ii (k) \\
w_\ii (k) \\
\end{bmatrix}\\
& \qquad\qquad\qquad~ +
\begin{bmatrix}
Z_\ii^\top B \\
W_\ii^\top B \\
\end{bmatrix} u(k) +
\begin{bmatrix}
Z_\ii^\top \\
W_\ii^\top \\
\end{bmatrix}d(k)\\
& \bar y_{\ii}(k) \!=\!
\begin{bmatrix}
c_\ii Z_\ii & 0_{1 \times (\nn-\nu_\ii)}\\
\end{bmatrix}
\begin{bmatrix}
z_\ii (k) \\
w_\ii (k) \\
\end{bmatrix}
+ n_{\ii} (k) + a_{\ii}(k).\!\!\!\!\!\!
\end{split}
\end{align}
By dropping the unobservable sub-state $w_\ii$ from \eqref{eq:decomp}, the observable quotient sub-system of \eqref{eq:decomp} is obtained as
\begin{align}[left = \PP_{\ii}^o: \empheqlbrace\,]
\begin{split} \label{eq:obs_sys}
& z_\ii (k+1) = S_\ii z_\ii (k) +
Z_\ii^\top B u(k) +Z_\ii^\top d(k)\\
& \bar y_{\ii}(k) = t_\ii z_\ii (k) + n_{\ii} (k) + a_{\ii} (k),
\end{split}
\end{align}
where $S_\ii := Z_\ii^\top A Z_\ii$, $t_\ii := c_\ii Z_\ii$, and the pair $(S_\ii, t_\ii)$ is observable.

Then, the partial observer $\OO_\ii$ is designed by a Luenberger observer for \eqref{eq:obs_sys} given in the following form:
\begin{align} \label{eq:observer}
\begin{split}
\OO_\ii :~ &{\hat z}_\ii (k+1) = F_\ii \hat z_\ii (k) + Z_\ii^\top B u(k) + L_\ii \bar y_{\ii}(k),
\end{split}
\end{align}
where the injection gain $L_\ii$ is chosen so that $F_\ii:=S_\ii - L_\ii t_\ii$ is Schur stable.
The dynamics of state estimation error $\tilde z_\ii := \hat z_\ii - z_\ii$ is governed by
\begin{align*} \label{eq:err_dyn11}
\begin{split}
\mathcal{F}_\ii:{\tilde z}_\ii (k\!+\!1) = F_\ii \tilde z_\ii (k) + L_\ii n_\ii (k) - Z_\ii^\top d(k) + L_\ii a_\ii (k)
\end{split}
\end{align*} 
whose solution becomes
\begin{equation} \label{eq:err_z}
\tilde z_\ii(k) = v_\ii (k) + e_\ii (k),
\end{equation}
where $v_\ii (k) := F_\ii^k \tilde z_\ii(0) + \sum_{j=0}^{k-1} F_\ii^{k-1-j} \left(L_\ii n_\ii (j) - Z_\ii^\top d (j) \right)$ and $e_\ii (k) := \sum_{j=0}^{k-1} F_\ii^{k-1-j} L_\ii a_\ii (j)$.
Here, the attack-induced estimation error vector $e_\ii(k)$ may have arbitrary values.
For all $k \ge 0$ and $\ii \in [\pp]$, there exist $\mu_F \ge 1$ and $0 < \beta < 1$ such that $\|F_\ii^k\|_2 \le \mu_F \beta^k$ since $F_\ii$ is Schur stable. 
In addition, for some $\mu_L$ and $\mu_Z$, it holds that $\|F_\ii^k L_\ii\|_2 \le \mu_L \beta^k$ and $\|F_\ii^k Z_\ii^\top\|_2 \le \mu_Z \beta^k$.
Then, one can easily show that
\begin{align} \label{eq:err_eq1}
\begin{split}
\|v_\ii(k)\|_2 \le \mu_F\|\tilde z_\ii(0)\|_2 \beta^k + w_{\max} \le v_{\max} (k),
\end{split}
\end{align}
where $w_{\max} := (\mu_L n_{\max} + \mu_Z d_{\max})/(1-\beta)$ and $\displaystyle v_{\max} (k) := \max_{\ii\in[\pp]} \left\{\mu_F\|\tilde z_\ii(0)\|_2 \beta^k + w_{\max}\right\}$.
As $k$ increases, $v_{\max} (k)$ converges to $w_{\max}$.

%%%%%%%%%%%%%%%%%%%%%%
\subsection{Design of the Decoder} \label{subsec:dec}
%%%%%%%%%%%%%%%%%%%%%%

The decoder collects all the data $\hat z_\ii$ from the partial observers $\OO_\ii$ and formulates the problem in the form of \eqref{eq:static_y}.
To this end, $Z_\ii^\top x = z_\ii$ in \eqref{eq:sim_trans}, is used.
Appending $\nn-\nu_\ii$ zero row vectors, $0_{1\times \nn}$, to each $Z_\ii^\top$ and stacking them all, we have
\begin{equation} \label{eq:Zxzz}
\begin{bmatrix}
{Z_1^\nn}^\top  \\
\vdots \\
{Z_\pp^\nn}^\top \\
\end{bmatrix} x (k) =
\begin{bmatrix}
{z_1^\nn} (k) \\
\vdots \\
{z_\pp^\nn} (k) \\
\end{bmatrix} =
\begin{bmatrix}
{\hat z_1^\nn} (k) \\
\vdots \\
{\hat z_\pp^\nn} (k) \\
\end{bmatrix} -
\begin{bmatrix}
{\tilde z_1^\nn} (k) \\
\vdots \\
{\tilde z_\pp^\nn} (k) \\
\end{bmatrix},
\end{equation}
where
\begin{alignat}{2} \label{eq:aug_zero}
\begin{aligned}
{Z_\ii^\nn}^\top &:= \begin{bmatrix}
{Z_\ii}^\top \\
O_{(\nn-\nu_\ii)\times \nn} \\
\end{bmatrix}, \quad
& z_\ii^\nn(k) := \begin{bmatrix}
z_\ii(k) \\
0_{(\nn-\nu_\ii)\times 1} \\
\end{bmatrix},\\
\hat z_\ii^\nn(k) &:= \begin{bmatrix}
\hat z_\ii(k) \\
0_{(\nn-\nu_\ii)\times 1} \\
\end{bmatrix},
&\tilde z_\ii^\nn(k) := \begin{bmatrix}
\tilde z_\ii(k) \\
0_{(\nn-\nu_\ii)\times 1} \\
\end{bmatrix}.
\end{aligned}
\end{alignat} 
This augmentation of zeros is to match the size of each matrix so that it agrees with the $\nn$-stacked vector considered in Section \ref{sec:err}.
Now, \eqref{eq:Zxzz} with \eqref{eq:err_z} is written in a compact form as
\begin{align} \label{eq:dyn_hatz}
\hat z(k) = \Phi x(k) + \tilde z(k) = \Phi x(k) + v(k) + e(k) \in\R^\np,
\end{align}
where $
\Phi :=\left[Z_1^\nn~Z_2^\nn~\cdots~Z_\pp^\nn\right]^\top
$.
It is supposed that additional zero elements are also appended to $v_\ii (k)$ and $e_\ii (k)$ as in \eqref{eq:aug_zero}. 

Since \eqref{eq:dyn_hatz} exactly matches with \eqref{eq:static_y}, one can directly apply the error correction technique developed in Section \ref{sec:err} into \eqref{eq:dyn_hatz}. 
Theorems \ref{thm:finite_thmv} and \ref{thm:suff_qv} are mainly employed so as to recover $x(k)$.
Before applying them, one should check that three conditions on those theorems are satisfied for the given system \eqref{eq:plant}: boundedness of $v(k)$, $\qq$-sparsity of $e(k)$, and $\qq$-error correctability of $\Phi$.
The first two conditions are easily satisfied by Assumptions \ref{ass:bdd} and \ref{ass:sparse}.
That is, the noise vector $v_\ii(k)$ is bounded by $v_{\max} (k)$ for all $\ii \in [\pp]$ by \eqref{eq:err_eq1}, which is induced from Assumption \ref{ass:bdd}. 
Since the error vector $e_\ii (k)$ depends only on the attack element $a_\ii (j)$ for $0\le j\le k-1$ and $a(j)$ is $\qq$-sparse according to Assumption \ref{ass:sparse}, the vector $e(k)$ is ($\nn$-stacked) $\qq$-sparse.
An additional assumption should be declared for the last condition, the $\qq$-error correctability of $\Phi$, to be fulfilled.

% % % Assumption: ass: 2q redun
\begin{assmpt} \label{ass:redun}
	The pair $(A, C)$ is 2$\qq$-redundant observable.~\hfill$\Diamond$
\end{assmpt}

Under Assumption \ref{ass:redun}, Propositions \ref{lem:corr_equiv} and \ref{lem:red_obs_equiv} ensure the $\qq$-error correctability of $G$ where $G$ is given in \eqref{eq:barG}. 
Thus, $\Phi$ is also $\qq$-error correctable because $\RR(\Phi_{\Lambda^\nn}^\top)=\RR(G_{\Lambda^\nn}^\top)$ for any $\Lambda \subset [\pp]$ by the construction of $\Phi$ and its elements ${Z_\ii}^\top$.
Therefore, all three conditions on Theorems \ref{thm:finite_thmv} and \ref{thm:suff_qv} hold.

\begin{figure} [t]
	\centering
	\includegraphics[width=0.9\columnwidth]{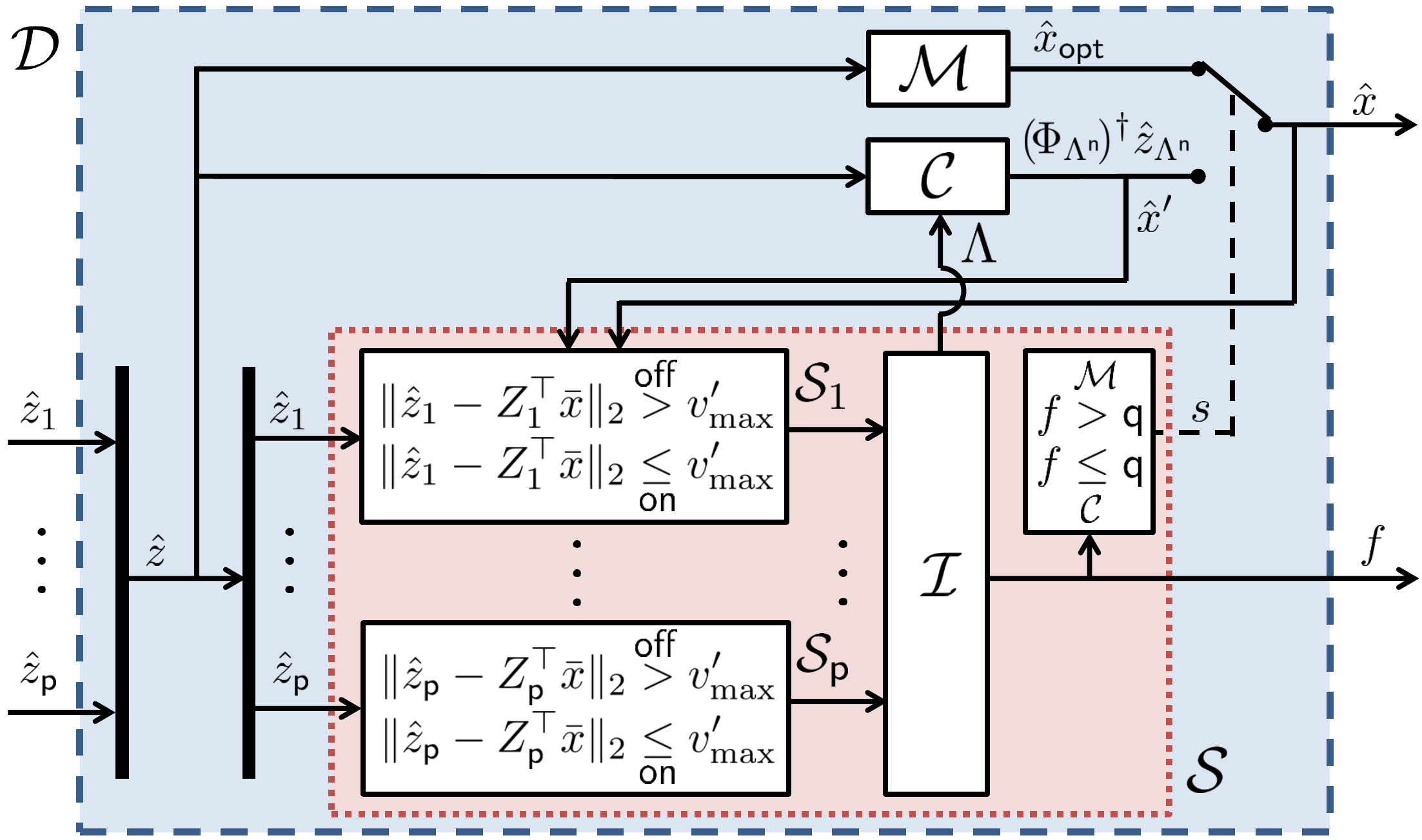}
	\caption{Configuration of the decoder.}\label{fig:decoder}
\end{figure}

%\medskip
\begin{algorithm} [t]
	\caption{Operation of the decoder}\label{alg:dec}
	\begin{algorithmic}[1]
		\REQUIRE $\hat z_1$, $\hat z_2$, $\cdots$, $\hat z_\pp$
		\ENSURE $\hat x$, $f$
		\INITIALIZE $\Lambda = [\pp]$ 
		\STATE {\bf while} System \eqref{eq:plant} is running {\bf do}
		\STATE \quad $\hat x' = \left(\Phi_{\Lambda^{\nn}}\right)^{\dagger}\hat z_{\Lambda^{\nn}}$
		\STATE \quad $f = \pp - \left| \left\{\ii\in[\pp] : \|\hat z_\ii - Z_\ii^\top \hat x' \|_2 \le v_{\max}'\right\} \right|$		
		\STATE \quad {\bf if} $f \le \qq$ {\bf then}
		\STATE \quad\quad Switch $s$ selects the line from $\CC$
		\STATE \quad\quad Calculator $\CC$ sets $ \hat x = \hat x'$
		\STATE \quad{\bf else} 
		\STATE \quad\quad Switch $s$ selects the line from $\MM$
		\STATE \quad\quad Minimizer $\MM$ solves \eqref{eq:finite_opt_v'} and produces $\hat x = \hat x_{\sf opt}$ 
		\STATE \quad\quad Update $\Lambda = \left\{\ii\in[\pp] : \|\hat z_\ii - Z_\ii^\top \hat x \|_2 \le v_{\max}'\right\}$ 
		\STATE \quad{\bf end if}
		\STATE {\bf end while}
	\end{algorithmic}
\end{algorithm}

The decoder's configuration is sketched in Fig.~\ref{fig:decoder} and its operation is described in Algorithm \ref{alg:dec}.
During the operation of the decoder, the monitoring scheme that the selector $\mathcal{S}$ and the switch $s$ perform is running on the basis of Theorem \ref{thm:suff_qv}, whereas the calculator $\CC$ and the minimizer $\MM$ have their roots in Theorems \ref{thm:detec_thmv} and \ref{thm:finite_thmv}, respectively. 
Note that $\bar x$ of each selector $\mathcal{S}_\ii$ in Fig.~\ref{fig:decoder} represents $\hat x'$ at the stage of monitoring (line 3 of Algorithm \ref{alg:dec}) and $\hat x$ during the updating step (line 10 of Algorithm \ref{alg:dec}).
If $f \le \qq$, the successful state estimation is ensured by Theorem \ref{thm:suff_qv}.(i). 
More specifically, we have $\|\hat x - x\|_2 \le \kappa_{\pp,\qq,\rr}^c (\Phi) v_{\max}$.
In this case, the index set $\Lambda$ can be supposed to be attack-free, and hence the calculator $\CC$ can recover the original state $x$ approximately by $ \hat x = \left(\Phi_{\Lambda^{\nn}}\right)^{\dagger}\hat z_{\Lambda^{\nn}}$, which is attributed to Theorem \ref{thm:detec_thmv}.
On the other hand, if $f > \qq$, the state estimate $\hat x$ is not close enough to the original state $x$ by Theorem \ref{thm:suff_qv}.(ii).
Hence, the algorithm goes to the minimizer step (i.e., the switch $s$ chooses the side of the minimizer $\MM$) to figure out new healthy sensors and the state estimates $\hat x$ by $\hat x_{\sf opt}$.
Furthermore, Theorem \ref{thm:finite_thmv} guarantees that $\|\hat x - x\|_2 \le \kappa_{\pp,\qq,\rr}^c (\Phi) v_{\max}$.
These results are summarized in the following theorem.

% % % Theorem
\begin{thm} \label{thm:robust_secure}
	Under Assumptions \ref{ass:bdd}, \ref{ass:sparse}, and \ref{ass:redun}, the estimator $\EE$ equipped with the observers $\OO_\ii$ given by \eqref{eq:observer} and the decoder $\DD$ employing Algorithm \ref{alg:dec}, guarantees that
	\begin{equation*}
	\|\hat x(k) - x(k)\|_2 \le \kappa_{\pp,\qq,\rr}^c (\Phi)~v_{\max} (k), \quad {}^\forall k \ge 0,
	\end{equation*} 
	where $\lim_{k\to\infty} v_{\max}(k) = w_{\max}$.
\end{thm}

% % % Remark:
\begin{rem} 
	For the resilient state estimation, most of the computational burden originates from the process of solving the optimization problem.
	The proposed decoder reduces the computational effort by combining the attack detection mechanism with the optimization process.
	Algorithm \ref{alg:dec} only requires the minimization problem to be solved for a very short time interval when the attacker first attempts to inject false data so that the decoder has $f > \qq$ at that instant.
	On the other hand, the estimator works as if there is no attack and computes only one simple pseudoinverse of a matrix during normal operation when $f \le \qq$ is guaranteed.
\end{rem}

% % % Remark:
\begin{rem} 
	Other observer-based resilient state estimators such as those in \cite{Pasqualetti13} and \cite{Chong15}, consist of all possible combinations of estimator candidates.
	Thus, they need to run $\binom{\pp}{\qq}$ estimators so that the required memory size is $\nn\binom{\pp}{\qq}$. 
    On the other hand, the total memory size of all partial observers in the proposed estimator, $\sum_{\ii=1}^{\pp} \nu_{\ii}$, is not greater than $\np$ because the size of each partial observer $\OO_\ii$ is $\nu_\ii \le \nn$ for all $\ii\in[\pp]$.  
\end{rem}

%%%%%%%%%%%%%%%%%%%%%%
\section{Simulation Results: Three Inertia System} \label{sec:example}
%%%%%%%%%%%%%%%%%%%%%%

\begin{figure}[t]
	\centering
	\begin{subfigure}[t]{.6\columnwidth}
		\centering
		\includegraphics[width=\columnwidth]{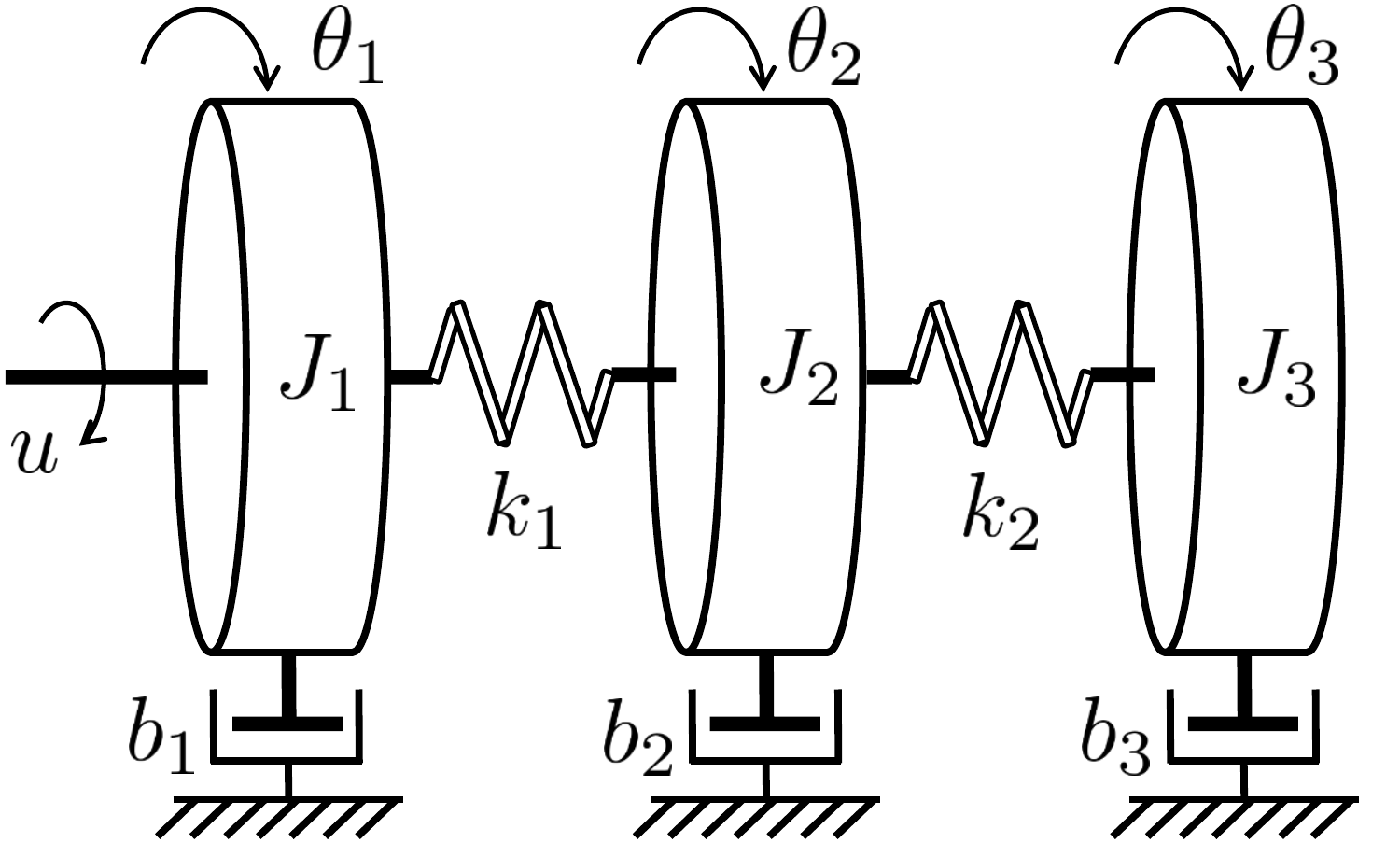}
		\caption{System model}
		\label{fig:example}
	\end{subfigure}
    
	\begin{subfigure}[t]{.75\columnwidth}
		\centering
		\includegraphics[width=\columnwidth]{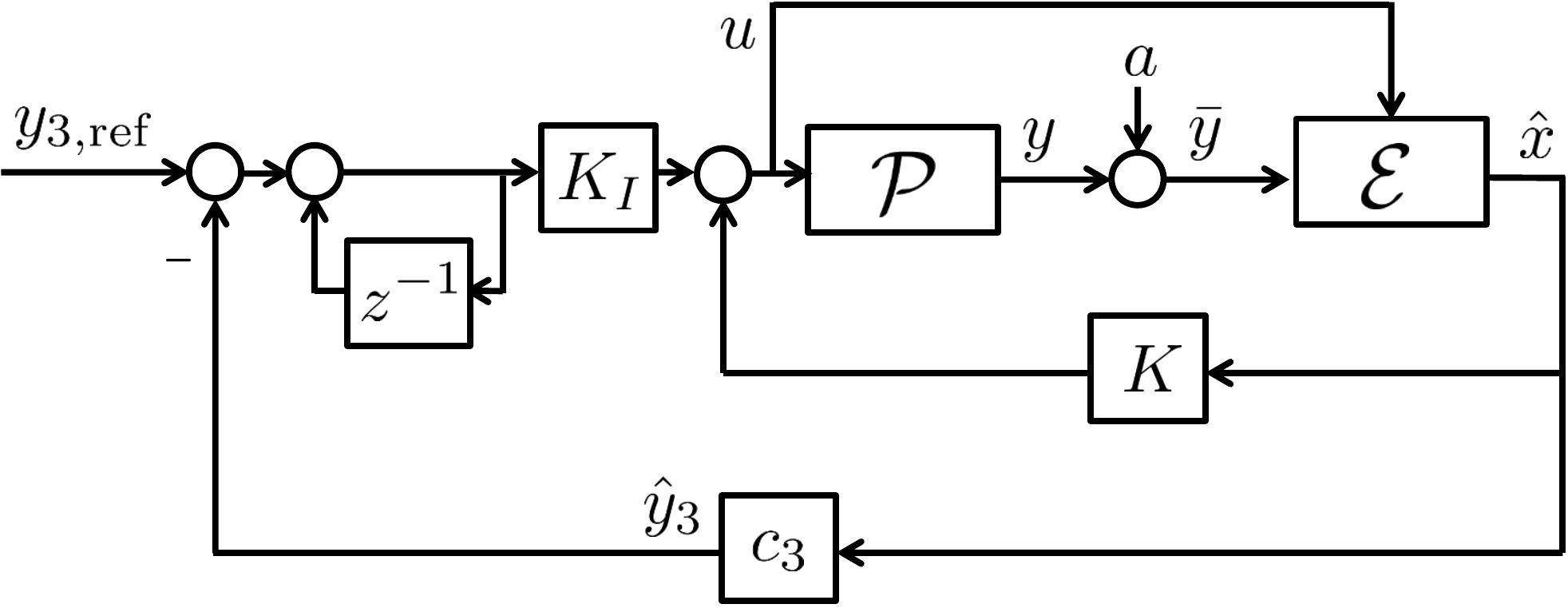}
		\caption{Control block diagram}
		\label{fig:int}
	\end{subfigure}
	\caption{Three inertia system and its control scheme.}
	\label{fig:exsystem}
\end{figure}

To verify the effectiveness of the proposed scheme, simulations with a three inertia system are conducted in this section.
The configuration of the three inertia system is described in Fig.~\ref{fig:example} and its dynamics can be represented by a continuous-time state-space equation 
\begin{align}[left = \PP_c: \empheqlbrace\,] 
\begin{split}\label{eq:plantc}
& \dot x(t) = A_c x(t) + B_c u(t) + d(t)\\
& y(t) = C_c x(t) + n(t) + a(t)
\end{split}
\end{align}
with the matrices
\begin{align*} 
\begin{split}
A_c&=\begin{bmatrix}
0 & 1 & 0 & 0 & 0 & 0 \\
-\frac{k_1}{J_1} & -\frac{b_1}{J_1} & \frac{k_1}{J_1} & 0 & 0 & 0 \\
0 & 0 & 0 & 1 & 0 & 0 \\
\frac{k_1}{J_2} & 0 & -\frac{k_1 + k_2}{J_2} & -\frac{b_2}{J_2} & \frac{k_2}{J_2} & 0 \\
0 & 0 & 0 & 0 & 0 & 1 \\
0 & 0 & \frac{k_2}{J_3} & 0 & -\frac{k_2}{J_3} & -\frac{b_3}{J_3} \\
\end{bmatrix},
\end{split}\\ 
\begin{split}
B_c&=\begin{bmatrix}
0 \\
\frac{1}{J_1} \\
0 \\
0 \\
0 \\
0 \\
\end{bmatrix},\quad
C_c=\begin{bmatrix}
1 & 0 & 0 & 0 & 0 & 0 \\
0 & 0 & 1 & 0 & 0 & 0 \\
0 & 0 & 0 & 0 & 1 & 0 \\
1 & 0 & -1 & 0 & 0 & 0 \\
0 & 0 & 1 & 0 & -1 & 0 \\
\end{bmatrix},
\end{split}
\end{align*}
where $J_1=J_2=J_3=0.01$ $\rm{kg\!\cdot\! m^2}$, $b_1=b_2=b_3=0.007$ $\rm{N\!\cdot\! m/(rad/s)}$, and $k_1=k_2=1.37$ $\rm{N\!\cdot\! m/rad}$.
Here, the state variables are $x:=[\theta_1~~\dot\theta_1~~\theta_2~~\dot\theta_2~~ \theta_3~~\dot\theta_3]^\top$ and the output measurements are $y:=\left[\theta_1~~\theta_2~~\theta_3~~\theta_1\!-\!\theta_2~~ \theta_2\! -\!\theta_3\right]^\top$.
In addition, the plant is corrupted by the uniformly bounded process disturbance $d$ and measurement noise $n$ with $d_{\max} = n_{\max} = 0.001$. 
To conduct a discrete-time simulation, the zero-order hold equivalent model of \eqref{eq:plantc} is considered, that is, the matrices of the discrete-time system \eqref{eq:plant} are given by
$A := e^{A_c T_s}$, $B := \Big( \int_0^{T_s} e^{A_c \tau}d\tau \Big) B_c$, and $C := C_c$
where $T_s := 1$~ms denotes the sampling time.
Note that the pair $(A, C)$ is 2-redundant observable, which implies that one can correct the 1-sparse attack signal and its dynamic security index becomes $3$.  
The control objective is to make the output $\theta_3$ follow the step reference $\theta_{3,{\rm ref}}$.
To this end, an observer-based feedback integral control scheme is adopted, as illustrated in \cite[Section 6.7]{Ogata95} and also in Fig.~\ref{fig:int}.
First, the state feedback gains $K$ and $K_I$ are chosen as
$$K:=-[2.32~~ 0.25~~ -\!2.47~~0.04~~1.70~~0.12], ~~~ K_I:=0.002,$$
as if the state $x$ is available.
Then, instead of using the conventional Luenberger observer, the proposed estimator ${\cal E}$ provides the estimate $\hat x$ of $x$.
The injection gains $L_\ii$ of the partial observer \eqref{eq:observer} in ${\cal E}$ are chosen arbitrarily such that $F_\ii=S_\ii - L_\ii t_\ii$ is Schur stable.
Attack signals are illustrated in Fig.~\ref{fig:attackval}, which describes that adversaries launch a measurement data injection attack at $t=2$~s so that the first sensor is compromised.
Figures \ref{fig:t1} and \ref{fig:t2d} show state trajectories $\theta_1(t)$, ${\dot\theta}_2 (t)$, and their estimates.
It demonstrates the attack-resilient property of our estimation algorithm.
Finally, Fig.~\ref{fig:ref} shows the reference tracking performance of the proposed control scheme.

\begin{figure}[t]
	\centering
	\begin{subfigure}[t]{0.95\columnwidth}
		\centering
		\includegraphics[width=\columnwidth]{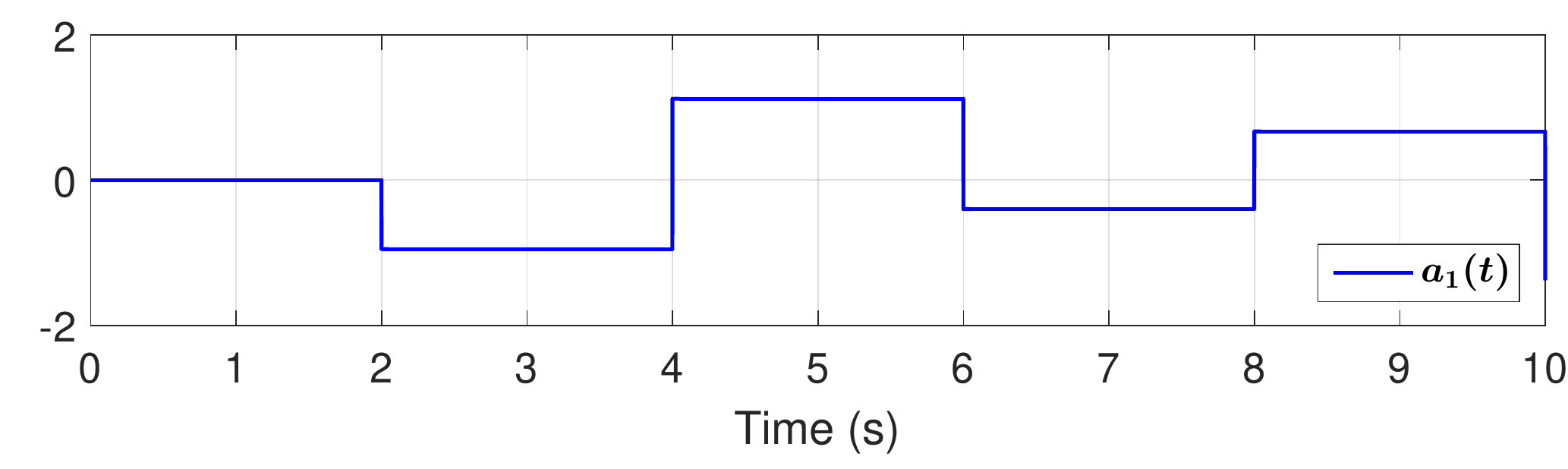}
		\caption{$a_1(t)$}\label{fig:attackval}
	\end{subfigure}
    
	\begin{subfigure}[t]{0.95\columnwidth}
		\centering
		\includegraphics[width=\columnwidth]{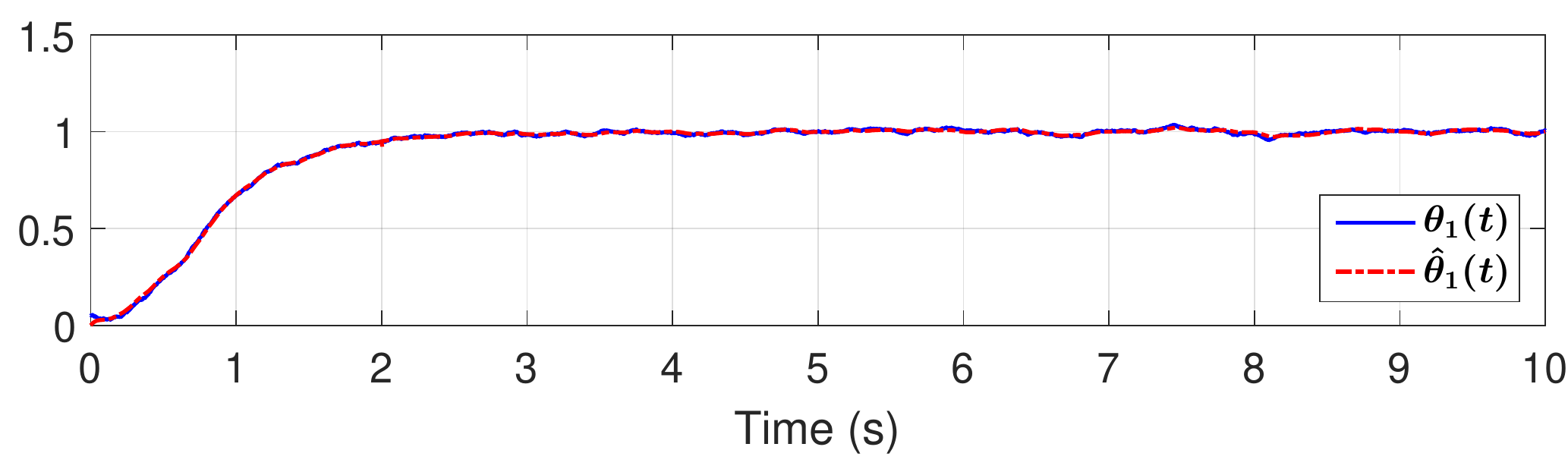}
		\caption{$\theta_1(t)$ and $\hat \theta_1(t)$}\label{fig:t1}
	\end{subfigure}
	
	\begin{subfigure}[t]{0.95\columnwidth}
		\centering
		\includegraphics[width=\columnwidth]{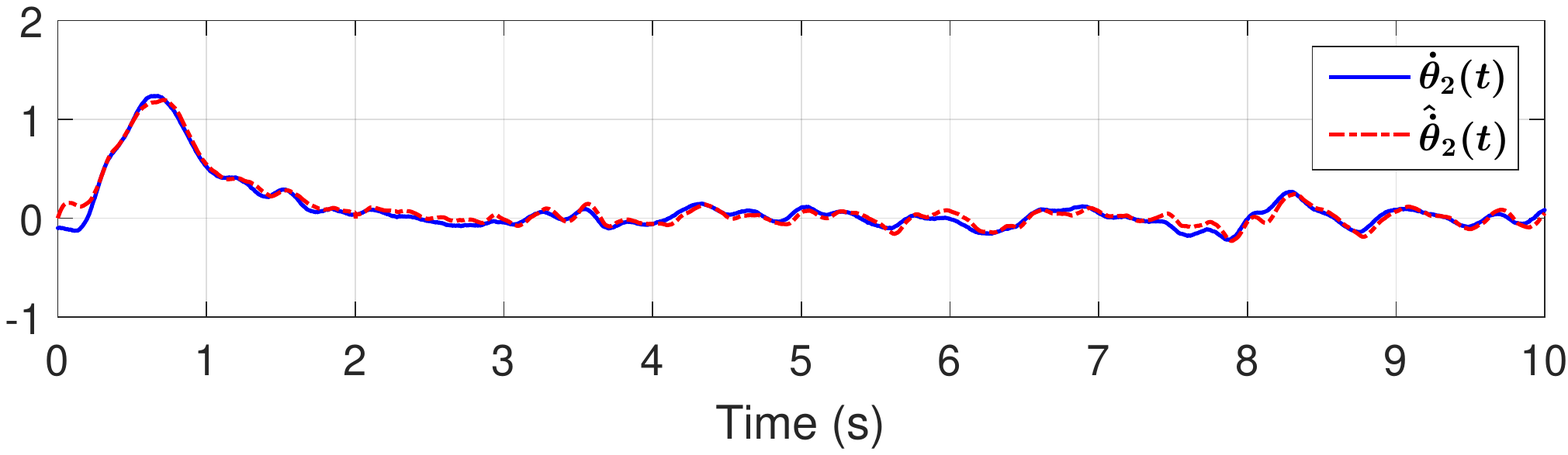}
		\caption{$\dot\theta_2(t)$ and $\hat{\dot\theta}_2(t)$}\label{fig:t2d}
	\end{subfigure}
    
	\begin{subfigure}[t]{0.95\columnwidth}
		\centering
		\includegraphics[width=\columnwidth]{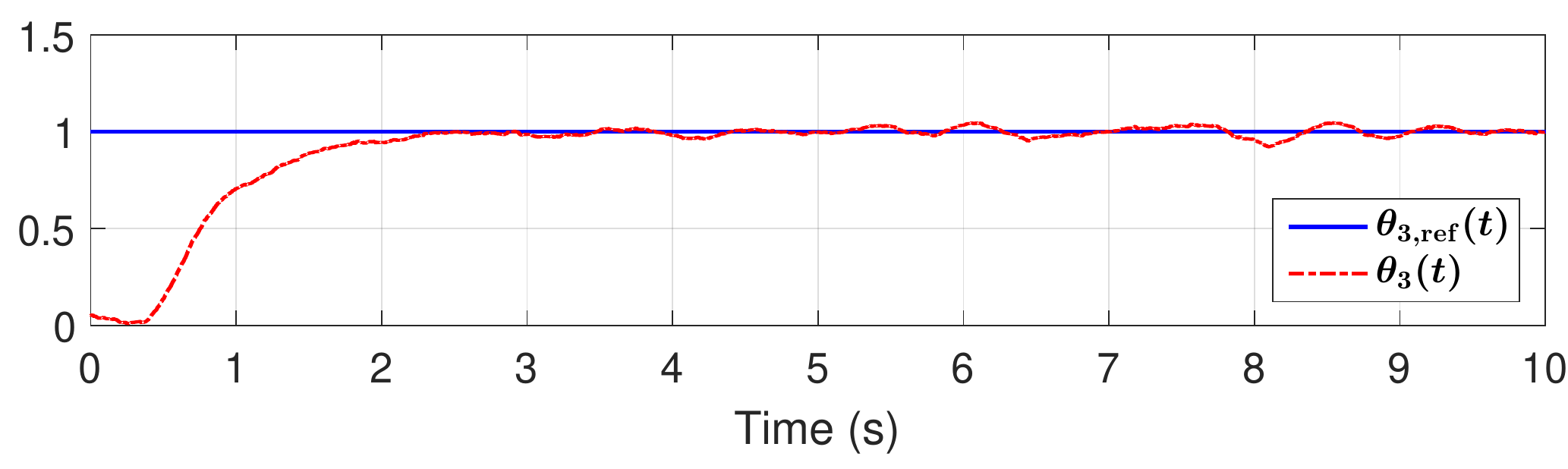}
		\caption{$\theta_{3,{\rm ref}}(t)$ and $\theta_3(t)$}\label{fig:ref}
	\end{subfigure}
	\caption{Plot of signals.}\label{fig:signal}
\end{figure}

%%%%%%%%%%%%%%%%%%%%%%
\section{Conclusion} \label{sec:con}
%%%%%%%%%%%%%%%%%%%%%%

An LTI system is said to be $2\qq$-redundant observable if it is observable even after eliminating any $2\qq$ measurements.
Relationships between the redundant observability and the security problems on cyber-physical systems under sensor attacks have been examined.
To summarize, $2\qq$-redundant observability implies that the numbers of detectable and correctable sensor attacks are $2\qq$ and $\qq$, respectively. 
In addition, the dynamic security index, the minimum number of attacks to remain undetectable, is $2\qq + 1$.

Assuming that the measurement data injection attack is $\qq$-sparse and the disturbances/noises are bounded, an attack-resilient and robust state estimation scheme has been proposed under $2\qq$-redundant observability.
The proposed estimator consists of a bank of partial observers operating based on the Kalman observability decomposition and a decoder exploiting error correction techniques.
In terms of time complexity, the decoder reduces the required computational effort by reducing the search space to a finite set and by combining a detection algorithm with the optimization process.
On the other hand, in terms of space complexity, the required memory is linear with the number of sensors by means of the decomposition used for constructing a bank of partial observers.

%%%%%%%%%%%%%%%%%%%%%%

\end{document}